\documentclass[%
 prl, 
 twocolumn,
 amsmath,amssymb,
 aps,
]{revtex4}

\usepackage{graphicx}
\usepackage{dcolumn}
\usepackage{booktabs}

\usepackage{bm}

\usepackage{xcolor}
\usepackage{xr}
\usepackage{comment}
\usepackage{hyperref}
\usepackage{cleveref}
\usepackage{enumitem}
\usepackage{accents}
\hypersetup{colorlinks,linkcolor={blue!90!black},citecolor={blue!90!black},urlcolor={blue!90!black}}

\DeclareMathAlphabet\mathbfcal{OMS}{cmsy}{b}{n}

\usepackage{amsmath}
\usepackage{mathtools}

\usepackage{amsmath}
\usepackage{amsthm}

\DeclareMathOperator{\tr}{tr}

\newcommand{\cdfsim}{\overset{\text{cdf}}{\ \sim\ }}
\newcommand{\pdfsim}{\overset{\text{pdf}}{\ \sim\ }}
\newcommand{\dconv}{\overset{\text{d}}{\to}}
\newcommand{\Beta}{\mathrm{B}}
\newcommand{\R}{\mathbb{R}}
\newcommand{\Prb}{\mathrm{P}}
\newcommand{\norm}[1]{\left\lVert#1\right\rVert}
\newcommand{\SP}{\mathbb{S}}
\newcommand{\TP}{\mathbb{T}}
\newcommand{\E}{\mathbb{E}}
\newcommand{\N}{\mathbb{N}}
\newcommand{\evalat}{\big\rvert}
\newcommand{\magn}[1]{\left|#1\right|}

\newtheorem{theorem}{Theorem}[section] 
\newtheorem{lemma}[theorem]{Lemma} 
\newtheorem*{corollary}{Corollary} 
\theoremstyle{definition}
\newtheorem{definition}{Definition}[section] 
\newtheorem{example}{Example}[section] 

\begin{document}


\title{Stochastic Voronoi Tessellations as Models for\\ Cellular Neighborhoods in Simple Multicellular Organisms}


\author{Anand Srinivasan}
\email[]{as3273@cam.ac.uk}
\affiliation{Department of Applied Mathematics and Theoretical 
Physics, Centre for Mathematical Sciences,\\ University of Cambridge, Wilberforce Road, Cambridge CB3 0WA, 
United Kingdom}%
\author{Steph S.M.H. H{\"o}hn}
\email[]{sh753@cam.ac.uk}
\affiliation{Department of Applied Mathematics and Theoretical 
Physics, Centre for Mathematical Sciences,\\ University of Cambridge, Wilberforce Road, Cambridge CB3 0WA, 
United Kingdom}
\author{Raymond E. Goldstein}
\email[]{R.E.Goldstein@damtp.cam.ac.uk}
\affiliation{Department of Applied Mathematics and Theoretical 
Physics, Centre for Mathematical Sciences,\\ University of Cambridge, Wilberforce Road, Cambridge CB3 0WA, 
United Kingdom}%

\date{\today}

\begin{abstract}  Recent work on distinct multicellular organisms
has revealed a hitherto unknown type of biological noise; rather than a regular 
arrangement, cellular  neighborhood volumes, obtained by Voronoi tessellations of the cell 
locations, are broadly distributed and consistent with gamma distributions.  
We propose an explanation for those observations in the case of the alga {\it Volvox}, whose 
somatic cells are embedded in an extracellular matrix (ECM) they export.  Both a solvable
one-dimensional model of ECM growth derived from bursty transcriptional activity and a two-dimensional 
``Voronoi liquid" model are shown to provide
one-parameter families that smoothly interpolate between the empirically-observed 
near-maximum-entropy gamma distributions and the crystalline limit of Gaussian distributions 
governed by the central limit theorem.   
These results highlight a universal consequence of intrinsic biological noise on the 
architecture of certain tissues.

\end{abstract}
\maketitle

Some of the simplest multicellular organisms consist of tens, 
hundreds, or thousands of cells arranged in an extracellular matrix (ECM), a network of 
proteins and biopolymers secreted by the cells.
They often have a simple geometry: linear chains and rosettes of 
choanoflagellates \cite{Dayel,Fairclough}, sheets and spheres of  
algae \cite{Hallmann}, and cylinders of sponges \cite{sponges}.   
While at first glance the arrangement of cells in the ECM appears regular, 
recent work \cite{Dayetal} revealed a hitherto undocumented disorder found by assigning neighborhoods to cells 
through a Voronoi tessellation based on cell centers.
Strikingly, both the lab-evolved ``snowflake yeast" \cite{Ratcliff} (a ramified 
form found after rounds of selection for sedimentation speed) and the alga 
{\it Volvox carteri} have broad distributions of Voronoi volumes $v$ accurately fit by $k$-gamma distributions
\begin{equation}
    p(v)=\frac{1}{\bar{v}-v_c}\frac{k^kx^{k-1}}{\Gamma(k)}
    \exp(-kx),\quad x=\frac{v-v_c}{\bar{v}-v_c},
    \label{eq:translated_gamma}
\end{equation}
where $\overline{v}$ is the mean volume and $v_c$ is the cell size. 
Particularly for \textit{Volvox}, these observations are central to a general question in developmental biology: 
{\it How do cells produce structures external to themselves in an accurate and robust manner?}

{\it Volvox} is one of the simplest multicellular systems 
with which to study statistical fluctuations in ECM generation.  
The adult (Fig. \ref{fig1}(a)) consists of $\sim\! 10^3$ somatic cells embedded at 
the surface of a transparent ECM, the uppermost layer of which 
is a thin elastic shell $\sim\! 500\,\mu$m 
in diameter and $\sim\! 30\,\mu$m thick, with a more gelatinous interior below; the organism is 
$>98 \%$ ECM. 
Daughter colonies develop from germ cells below the outer layer through rounds of 
binary division that yield a raft of cells held together by cytoplasmic bridges 
remaining
after incomplete cytokinesis.  
Following ``embryonic inversion" that turns the raft inside-out \cite{inversion}, 
daughters enlarge by export of ECM proteins, 
expanding the colony to its final size over the course of a day, during which the 
widely-distributed neighborhood volumes appear.  
Figure \ref{fig1}(a) shows a a section of 
the Voronoi tessellation obtained by light-sheet imaging \cite{Dayetal}. 
The area distribution of Voronoi partitions across $6$ organisms is shown in 
Fig. \ref{fig1}(b) along with a fit of the gamma distribution 
\eqref{eq:translated_gamma} that yields $k\approx 2.35 \pm 0.04$ (95\% CI). 

The general issue above becomes the question of how cells generate ECM 
so that the spheroidal form is maintained during the dramatic colony growth despite the 
strong right-skew of the neighborhood volume distribution \eqref{eq:translated_gamma}. 
A biological answer might invoke cell signaling in response to mechanical forces as a 
mechanism to coordinate growth and would ascribe the distribution 
\eqref{eq:translated_gamma} to imperfections 
in that process.  
Surprisingly, the novel problem of neighborhood distributions is so little-studied 
that we do not even understand quantitatively the feedback-free case, surely a 
benchmark for any analysis of correlations.  
Work in granular physics \cite{AsteEntropy} has shown that \eqref{eq:translated_gamma} 
arises from maximizing entropy of partitions 
subdividing a volume subject to a fixed mean partition size.  
But this begs the question of why biological systems should follow a maximum-entropy principle. 

\begin{figure}[t]
\includegraphics[width=\columnwidth]{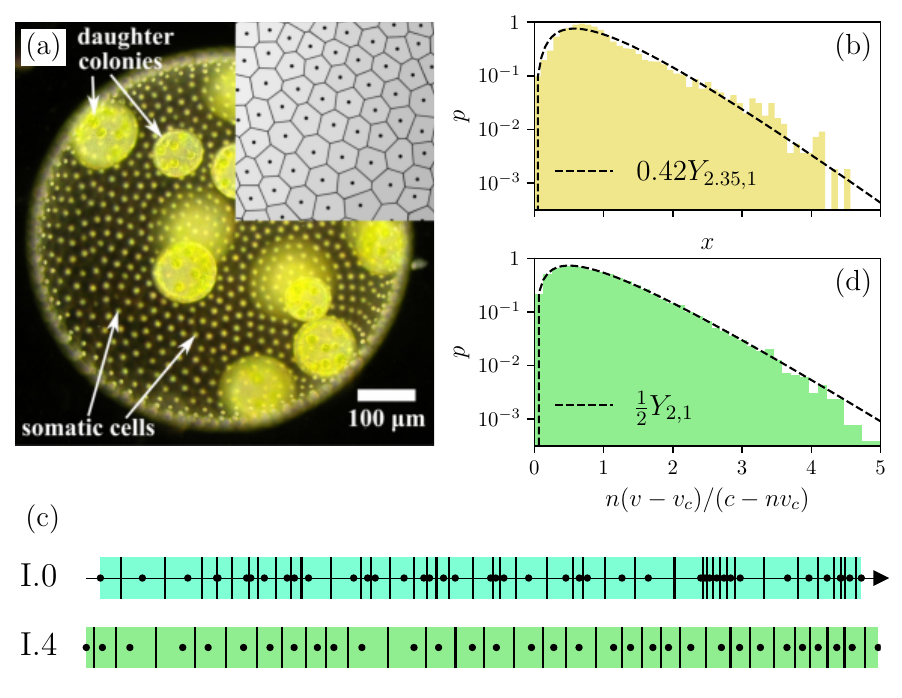}
\caption{The green alga {\it Volvox} and 1D models for cell positions.
(a) Adult with cell types labelled, and section of Voronoi 
tessellation around somatic cells (dots within each Voronoi polygon), adapted from \cite{Dayetal}. 
(b) Voronoi areas $x$ follow the translated gamma distribution 
\eqref{eq:translated_gamma}, computed by a fit of $k$.  
(c) Model \protect\hyperlink{I.1}{I.1} \cite{SM}, the Poisson process on $[0, \infty)$, where cells are indicated by dots, Voronoi boundaries by vertical 
line segments and intercalating ECM is colored.
Model \protect\hyperlink{I.4}{I.4}, the circular Poisson process with minimum spacing.
(d) Numerical distribution for model \protect\hyperlink{I.4}{I.4}, with no cell overlaps, 
compared to the analytical gamma distribution for large $n$.
}
\label{fig1}
\end{figure}

Here we study perhaps the simplest models for cellular 
positioning within a thin ECM, where noisy matrix production by 
statistically identical cells
causes them to space apart randomly during growth.
We formulate the resulting stochastic cellular configuration as a 
\textit{point process} \cite{DaleyPointProcess} whose Voronoi tessellation is a 
well-studied topic in stochastic geometry \cite{StochasticGeometry} and serves as a 
model for cellular neighborhoods.  
A class of analytically solvable 1D models is used to
illustrate how 
gamma distributions \eqref{eq:translated_gamma} may arise from 
feedback-free growth processes, and a one-parameter family of 
2D stochastic Voronoi tessellations is introduced as a prototype 
of systems with interactions between polygons.
The following is a non-technical summary; details are 
in Supplementary Material \cite{SM}.
In the following, capital letters $X_{i;\theta_1,\ldots}$ denote random variables 
(r.v.s) $i$ with parameters $\theta_1,\ldots$, and $W,X,Y,Z$ denote Gaussian, exponential, 
gamma, and beta r.v.s.

In an act of extreme reductionism, consider a linear \textit{Volvox} (termed 
model \hyperlink{I.0}{I.0} \cite{SM}).  
The gamma distribution has the property of being self-\textit{divisible}, that is, 
$Y_{1;k,\lambda} = Y_{2;\frac{k}{2},\lambda} + Y_{3;\frac{k}{2},\lambda}$ for 
independent $Y_2, Y_3$ by the convolution rule for sums of random variables \cite{SM}.
In a 1D ECM (Fig. \ref{fig1}c), if $L_i$ are the cellular spacings then 
$V_i = (L_i + L_{i+1})/2$ are its Voronoi segments, which suggests the 
decomposition of $V_i$ into spacings $L_i$ which are themselves independent and 
identically distributed (i.i.d.) gamma r.v.s.
Similarly, $L_i$ is itself interpretable as arising from i.i.d. gamma-distributed mass increments during growth.
This leads to two classes (i, ii) of configurations which one may expect to observe experimentally.
The first (i) is that where $L_i$ is formed by a large number $k$ of small random 
mass increments, where $k$-gamma 
converges to a Gaussian by the central limit theorem, and at fixed mean to a delta distribution \cite{SM}. 

The second (ii) is the case in which $L_i$ is formed by a small number of 
large mass increments, suggesting some intermediate piece comprising the ECM is produced 
at low copy number.
A plausible precedent for fluctuations possessing this particular distribution 
is the bursting protein transcriptional activity observed in simple 
unicellular organisms such as \textit{E. coli} \cite{StochasticProtein1}. 
There, it is known that mRNA transcription occurs at some rate when a gene is 
turned on, individual mRNA molecules are transcribed at some rate into proteins before 
degrading (e.g. by RNases) with an exponentially-distributed lifetime, and 
individual protein bursts exported into the extracellular environment 
correspond 1-1 with individual mRNA transcripts within the cell.
Translating this phenomenology to {\it Volvox}, we hypothesize  
that $L_i \propto P_i$ where $P_i$ is the steady-state extracellular concentration of a 
protein $P$ governing ECM assembly. 
Then the time-dependent concentration $P_i(t)$ is a pure-jump process with some total number $b$ of 
exponentially-distributed bursts, resulting in the $b$-gamma-distributed spacings $L_i = Y_{i;b,\lambda}$.
(Fractional values of $b$, representing cross-cell averages, yield the same stationary 
$b$-gamma distribution, as shown from the master equation for $P_i(t)$ with 
exponential kernel \cite{StochasticProtein2}.
We take $b \ge 1$ as every cellular neighborhood 
grows in \textit{Volvox} without cell division.)
Then, $V_i$ is $2b$-gamma distributed,
\begin{equation}
    V_i = \frac{X_{j;\lambda}+X_{j+1;\lambda}}{2} =  \frac12Y_{i;2b,\lambda} \pdfsim \frac{4\lambda^{2b}v^{2b-1}e^{-2\lambda v}}{\Gamma(2b)},
    \label{eq:1d_poisson_gamma}
\end{equation}
which, apart from the offset $v_c$, is \eqref{eq:translated_gamma} with $k=2b$.
That $k\approx 2.4$ in \textit{Volvox}, approaching the 1D lower bound of 
$k=2$ and apparently falling in class (ii), is consistent with observations that ISG, a 
glycoprotein critical to the ECM organization, is transcribed over a period of 10 minutes, quite short compared to its accumulation in the extracellular matrix 
over timescales comparable to the 48h life cycle \cite{HallmannISG}.

In the low-copy number limit $b\to 1$, cellular positions 
$R_i = \sum_{j \le i}L_j$ occur as a Poisson process.
This is the maximum-entropy configuration, in which cell positions are uncorrelated 
in the sense that for any fixed number of cells $N$ occurring within a fixed segment of 
size $L$, $\{R_i\}_{i=1}^N$ are i.i.d. uniform random variables \cite{SM}.
Of course, the gamma distribution is supported on $[0, \infty)$ and one must consider finite-size 
effects.
It can be shown \cite{SM} that on a circular ECM of fixed circumference $C$ with fixed or variable 
cell count $N$ (termed models \hyperlink{I.1}{I.1}-\hyperlink{I.3}{I.3}), the marginal distribution 
of Voronoi lengths given the above  
converges in the large-$C,N$ limits at fixed cell number density to the same $2b$-gamma 
distribution\textemdash analogous to the convergence of ensembles in statistical 
physics in the thermodynamic limit. 

To complete the 1D analysis, we show that the offset $v_c$ in \eqref{eq:translated_gamma} from finite cell sizes.
Suppose cells with centers of mass at $\{R_i\}_{i=1}^n$ on a circle of circumference $c$ have uniform size $v_c$ with $nv_c \le c$, so that $L_i \ge v_c$ for all $i$.
As we are in 1D, this is expressible as $L_i = v_c + \Tilde{L}_i$, where $\Tilde{L}_i$ are the random spacings of a smaller circle of circumference $c - nv_c$.
This reduces to the fixed-$N, C$ case of model \hyperlink{I.2}{I.2}, hence we have the marginal Beta distribution for Voronoi lengths
\begin{align}
    V_i = v_c + \frac{c-nv_c}{2}Z_{i;2b,(n-2)b}.
\end{align}
Defining the cell number density within the remainder as $\rho = (n-2)/(c-nv_c)$ and taking the thermodynamic limit $n, c \to \infty$ with $\rho$ and $v_c$ fixed, we have by \eqref{eq:circle_gamma_convergence}
\begin{align}
    \label{eq:circle_gamma_vc}
    V_i &= v_c + \frac{n-2}{2\rho}Z_{i;2b,(n-2)b} \overset{n\to\infty}{\longrightarrow} v_c + Y_{i;2b,2\rho}.
\end{align}
Thus, the lengths $V_i$ with $2\rho:=\lambda$ in \eqref{eq:circle_gamma_vc} have precisely the 
distribution \eqref{eq:translated_gamma}
under the substitution $\lambda= k/ (\overline{v} - v_c)$.

Unlike in 1D where any sequence of partitions
forms the real line, cells in 2D are nearly always 
interacting since their neighborhoods are mutually constrained to be a
subdivision of the ECM.  Their positions are derived from the 
neighborhood configurations, as ECM is secreted during growth, and
we should expect that within a Voronoi description the cell locations
will depend on those partitions.
These geometric constraints co-exist with the possibility of 
maximum-entropy (Poisson) and minimum-entropy (crystalline) configurations.
The family of point processes we introduce below models cellular interactions 
based on their Voronoi tessellations, interpolates between these phases, and 
can be interpreted as arising from a strain energy in each neighborhood.
Our focus on geometry is complementary to recent work on topological properties of
tessellations \cite{Skinner}.

Let the ECM now be a bounded domain $\Omega \subset \R^2$, with area $\magn{\Omega} > 0$ and
a fixed number $n$ of cells. We 
assume that cells are scattered at positions 
$\{\bm{x}_i\}_{i=1}^n = \mathcal{X}$ with cellular neighborhoods $\{D_i\}_{i=1}^n = \mathcal{D}$ comprising $\Omega$ in a manner which minimizes 
\begin{align}
    \label{eq:quantization}
    E(\mathcal{X}, \mathcal{D}) &= \frac{1}{2}\sum_{i=1}^n\int_{D_i}\norm{\bm{x}-\bm{x}_i}^2d\bm{x},
\end{align}
Each summand of \eqref{eq:quantization} is the trace of $D_i$'s second area moment 
about $\bm{x}_i$, which for polygonal $D_i$ is the small-strain limit of the
bulk energy of a deformation 
from a regular $n$-gon centered at $\bm{x}_i$ \cite{E2}.
Alternatively, minimizing $E$ has
the interpretation of an optimal cell-placement 
problem with a cost for transporting resources produced by cells at 
$\bm{x}_i$ to other points $\bm{x}$ in the neighborhoods $D_i$.

\begin{figure*}[t!]
\includegraphics[width=1.99\columnwidth]{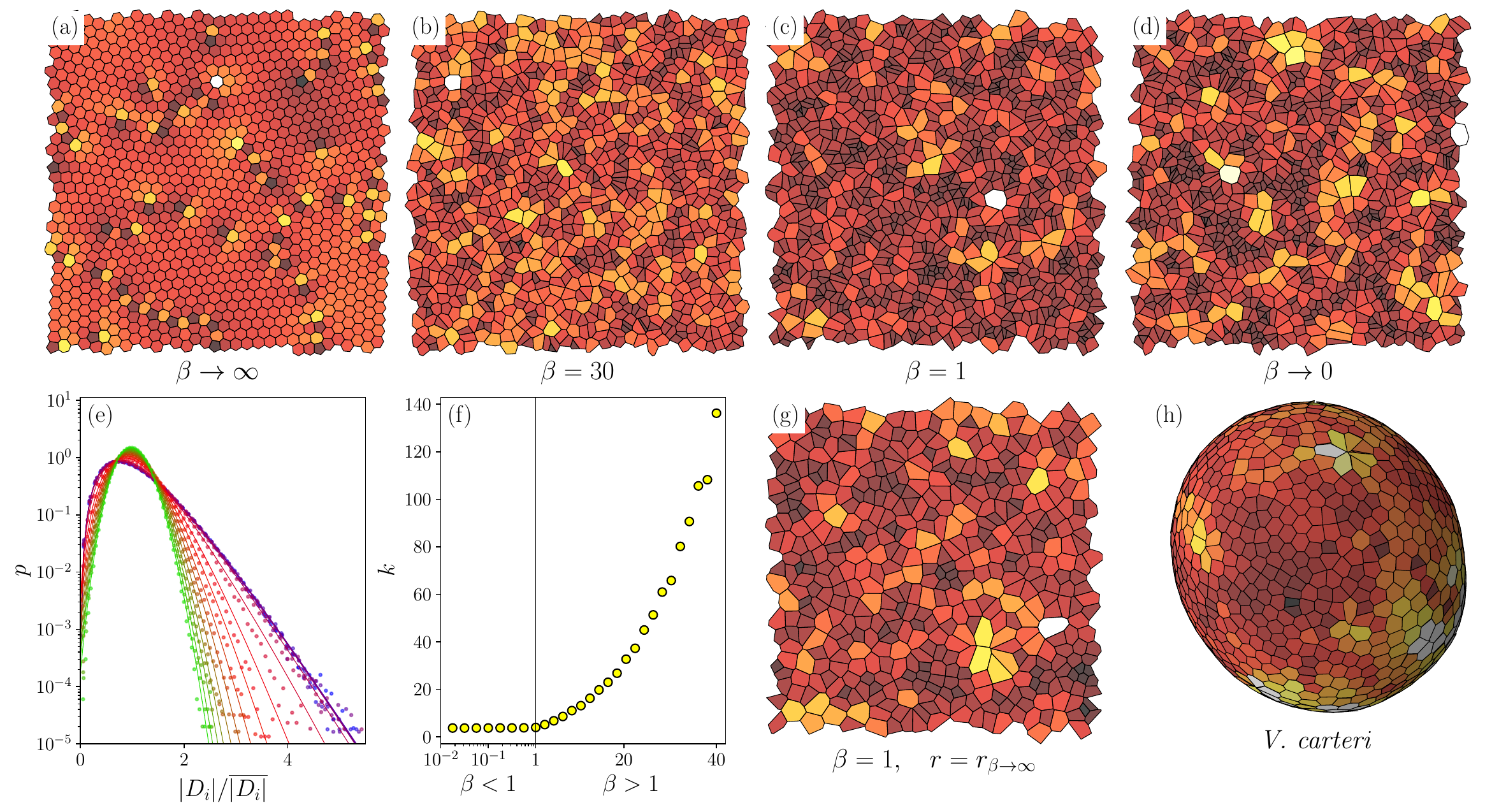}
\caption{Voronoi liquid interpolates between maximum- and minimum-entropy point configurations in 2D.
(a)-(d): Monte Carlo simulations of the Voronoi liquid at varying temperatures, from the infinite-temperature (Poisson) limit to the zero-temperature (sphere-packing) limit.
(e): Area distributions at all temperatures are approximately $k$-gamma distributed, with $k$ monotone increasing with $\beta$. 
(f): $k$ is constant for $\beta < 1$ and begins to grow superlinearly for $\beta > 1$.
(g) Voronoi tessellation of \textit{V. carteri} \cite{SM}.
(h): Minimum spacing $r$ enforced by the Mat\'ern thinning rule \cite{SM}, with $r$ equal to the minimum spacing in the frozen limit $\beta \to \infty$.  Voronoi
polygons are colored from black to white based on relative area in each panel.
}
\label{fig:pp_2d}
\end{figure*}

For fixed cell positions $\mathcal{X}$, the 
set $\mathcal{D}$ minimizing $E$ is precisely the Voronoi tessellation of $\mathcal{X}$.
To see this, note that for Voronoi $\mathcal{D}$, any other $\mathcal{D}'$, and  
a point $\bm{y}$ falling in $D_i\in \mathcal{D}$ and also in $D_i'\in \mathcal{D}'$, 
we have $\norm{\bm{y}-\bm{x}_i'} \ge \norm{\bm{y}-\bm{x}_i}$ by 
the Voronoi rule, hence $E(\mathcal{X},\mathcal{D}') \ge E(\mathcal{X},\mathcal{D})$.
Rescaling the coordinates $\bm{x} \mapsto \bm{x}\sqrt{\rho}$ to achieve unit number density $\rho = n/\magn{\Omega} \mapsto 1$, this motivates the study of the positional energy
\begin{equation}
    V(\mathcal{X}) = \rho^2E(\mathcal{X}, \text{Vor}(\mathcal{X})),
\end{equation}
where $\text{Vor}(\mathcal{X})$ is the Voronoi tessellation. 
For fixed neighborhoods $\mathcal{D}$, calculating $\partial E/\partial \bm{x}_i = 0$ 
shows that the minimizing positions $\mathcal{X}$ are the $\mathcal{D}$-centroids 
$\bm{\mu}_i = \magn{D_i}^{-1}\int_{D_i}\bm{x}d\bm{x}$; minimizers of $V$
are \textit{centroidal Voronoi tessellations} (CVTs), ubiquitous in meshing 
problems, clustering, and models of animal behavior \cite{CVT1}. 

Define a family of Gibbs point processes \cite{StochasticGeometry} 
whose joint 
positional distributions conditional on fixed $N$ are 
\begin{align}
    \label{eq:gibbs_voronoi}
    f_\beta(\mathcal{X}) &\propto \exp\left(- \beta V(\mathcal{X})\right),
\end{align}
indexed by a temperature-like quantity $\beta^{-1}$. 
Following others who have investigated phase transitions of this system 
\cite{VL1}, we refer to it as the \textit{Voronoi liquid}, which
differs from classical pair-potential fluids due to 
many-body interactions between Voronoi-incident particles.

The maximum-entropy case (ii)
is realized in the infinite-temperature limit $\beta\to 0$ with 
equiprobable configurations. 
This defines the ``Poisson-Voronoi tessellation" (PVT) \cite{StochasticGeometry}, 
which reduces to the exponentially-distributed spacings discussed in the 1D models above.
The areas $\magn{D_i}$
of 2D PVTs, a realization of which is shown in Fig. \ref{fig:pp_2d}(d), have 
been shown in numerical studies to conform to $k$-gamma 
distributions \cite{TanemuraVoronoi, WeaireVoronoi, FerencVoronoi}. 
A minimum-entropy configuration arises in the zero-temperature limit 
$\beta \to \infty$, where \eqref{eq:gibbs_voronoi} becomes degenerate and 
the configuration freezes to a hexagonal lattice as in Fig. \ref{fig:pp_2d}(d), which is the globally optimal CVT and densest sphere-packing in 2D \cite{CVT2}.
Prior approaches using a structure factor analysis \cite{Klatt} found that, by contrast, Lloyd iterations (corresponding to a ``fast quench'' at zero temperature \cite{VLQuench}) suppress crystalline configurations and adopt amorphous ``hyperuniform'' states.
We investigate now the finite-temperature range $\beta \in (0, \infty)$, and  
show that areas are accurately described by $k$-gamma distributions with $k$ an ``order parameter'' following a monotone relationship 
with $\beta$, analogous to the burst-count-driven spacing distributions of the 1D case.

As a generalization of our previous comment on nonuniqueness, in 2D the entropies of the Voronoi size distribution and of the positional distribution 
do not necessarily follow a monotone relation.
\textit{Volvox} itself (Fig. \ref{fig1}b) provides an example; its scaled 
area distribution \eqref{eq:translated_gamma} has $k\approx 2.3$, while Poisson-Voronoi 
tessellations of the flat torus and sphere have $k\approx 3.5$ \cite{SM}, yet their positional 
distribution is the maximum-entropy one.
Hence, ``entropy'' could refer to the differential entropy of its Voronoi size distribution 
\textit{or} to that of its joint distribution over positions at fixed $N$.

Since $V$ is $C^2$ \cite{CVT1}, \eqref{eq:gibbs_voronoi} is the stationary solution of a 
Langevin equation
\begin{align}
    \label{eq:langevin_vliq}
    dR_i(t) &= - \frac{\partial V}{\partial x_i}\biggr\rvert_{\{R_j(t)\}} dt + \sqrt{2\beta^{-1}} dW_i(t),
\end{align}
with $W_i(t)$ i.i.d. Brownian motions, and
time has been rescaled to $\beta^{-1} t$ to allow integration in the limit $\beta \to \infty$.
Since $\partial V / \partial \bm{x}_i \propto \magn{D_i}(\bm{x}_i - \bm{\mu}_i)$ 
\cite{SM}, \eqref{eq:langevin_vliq} may be interpreted as a neighborhood-centroid-seeking 
model of cellular dynamics during noisy growth or a Markov Chain Monte Carlo (MCMC) method to 
sample the stationary distribution \eqref{eq:gibbs_voronoi}.
An Euler-Maruyama discretization of \eqref{eq:langevin_vliq} does not satisfy 
detailed balance, but this can be rectified by adding
a Metropolis-Hastings step \cite{SM, RHMC}. Using this method, we investigated 
the statistical properties of the Voronoi liquid by numerically solving \eqref{eq:langevin_vliq}
for $n=10^3$ \cite{SM}, the somatic cell count of {\it Volvox}, in the
simplified geometry of a unit square with periodic boundary conditions to remove 
curvature and topology effects.

Figures \ref{fig:pp_2d}(a-d) show samples of the Voronoi liquid at varying temperatures with 
evident differences and similarities to \textit{Volvox}. 
As seen in Fig.~\ref{fig:pp_2d}(e,f), area distributions sampled at 13 logarithmically spaced 
values from $\beta=10^{-3}-1$ and 21 linearly spaced values from $1-40$, are well-described by $k$-gamma distributions with 
$k$ increasing monotonically with $\beta$.
This is consistent with the transition of $p(\magn{D_i}/\magn{\overline{D}_i})$ to 
a parabola on the log-scale in Fig. \ref{fig:pp_2d}e, the limit in which $k$-gamma 
approaches a Gaussian. It is in this sense that the control parameter $\beta$ 
is analogous to the protein burst count $b$ in 1D.
A similar monotone relationship between the ``granular temperature'' $\beta_{\text{gr}}^{-1}$ of a packing and $k$, in which partition size instead played the role of energy, has been noted previously in granular physics \cite{AsteEntropy, Powders}.

The importance of intermediate-entropy configurations is perhaps more readily seen in 2D. 
Studies of confluent tissue \cite{Atia} found that $k$-gamma distributions also arise in the aspect ratios (defined from the eigenvalues of the second area moment).
Poisson-Voronoi tessellations, notably, do not possess gamma-distributed aspect ratios. They instead follow an approximate beta-prime distribution, perhaps as a consequence of the gamma-distributed principal stretches \cite{SM}.
This is seen in Fig.~\ref{fig:pp_2d}(a), where high-aspect ratio ``shards'' occur at 
$\beta=0$, yet disappear at low temperature.
This raises questions of the underlying physics responsible for aspects of stochastic geometry than size.
As a simple extension, the Voronoi liquid with hard-sphere thinning \cite{SM} (one way to produce the offset $v_c$ \eqref{eq:translated_gamma} in 2D), a realization of which is shown 
in Fig.~\ref{fig:pp_2d}(h), does not exhibit these artifacts and more closely resembles 
the regular arrangement observed in {\it Volvox}, with both gamma-distributed areas and 
aspect ratios \cite{SM}. 

A biological interpretation of the Voronoi assumption is that the polygonal boundaries of each
cellular neighborhood are the colliding fronts of isotropically produced ECM 
material exported from cells.
Inverting the typical modelling procedure by \textit{assuming} that the Voronoi rule holds 
at some temperature $\beta^{-1}$, one can infer the 
distributional parameters 
using standard maximum-likelihood estimators for Gibbs point processes \cite{GPP}.
From the estimated temperature, for example, one can invert the 
$k$-$\beta$ relationship by monotonicity to deduce the copy number of bursty rate-limiting 
steps in growth.
Such estimators are critical for elastic models of tissues, where noise in 
individual cellular configurations co-exists with stable geometric properties of the population.
The stochastic Voronoi models we have presented here reproduce aspects of configurational noise, such as 
the empirically observed gamma distributions, and simultaneously provide a formal framework\textemdash an ensemble\textemdash within which 
to infer features of random finite configurations of cells such as interaction strength and preferential geometry.

As a purely 
mathematical construct, a Voronoi tessellation makes no reference to 
microstructure around cells, and it thus plays a role for tissues 
analogous to the random walk model of 
polymers and the hard sphere model of fluids.
Yet, 
each \textit{Volvox} somatic cell sits within a polygonal ``compartment" whose 
boundaries are composed of
denser material within the larger ECM \cite{compartment}.  Dimly visible in 
brightfield microscopy, 
these compartments have recently been labelled
fluorescently \cite{Benni1}, enabling the simultaneous
motion tracking of cells and growth of compartments during development. 
The strong correlation observed \cite{Benni1} 
between the location of these compartment boundaries and the 
the associated Voronoi partitions 
will enable tests of the connection hypothesized here between properties of 
stochastic ECM generation at the 
single cell level and population-level statistics.

\begin{acknowledgments}
We are grateful to T. Day, P. Yunker, and W. Ratcliff for numerous discussions.  This work was supported in part
by the Cambridge Trust (AS), The John Templeton Foundation and 
Wellcome Trust Investigator
Grant 207510/Z/17/Z (SSMHH, REG). 

\end{acknowledgments}

%

\vfil
\eject
\clearpage

\begin{widetext}

\section*{Supplemental Material}
This file discusses analytical and computational details pertaining to the main text. 

\setcounter{equation}{0}
\setcounter{figure}{0}
\setcounter{table}{0}
\setcounter{page}{1}
\makeatletter
\renewcommand{\theequation}{S\arabic{equation}}
\renewcommand{\thefigure}{S\arabic{figure}}
\renewcommand{\thedefinition}{S\arabic{definition}}
\renewcommand{\thetheorem}{S\arabic{theorem}}
\renewcommand{\theidentity}{S\arabic{identity}}
\renewcommand{\thelemma}{S\arabic{lemma}}

\bigskip

\section{Background on random variables}
Many of the following are standard facts, recalled for a self-contained reference.

\subsection{Transforms of random variables and convergence in distribution}
\label{sec:rv_basics}

\begin{definition}[Pushforward measure]
Let $g : \mathbb{R}^n\to \mathbb{R}^n$ be a diffeomorphism and $Y = g(X)$. The pushforward probability measure $\mu_Y$ is, for all measurable $U \subset \R^n$,
\begin{align}
    \mu_Y(U) &= \mu_X(g^{-1}(U)).
\end{align}
\end{definition}

\paragraph{Transforms of random variables.} 
When $\mu_X$ has a Radon-Nikodym derivative $f_X$ with respect to the Lebesgue measure \textemdash i.e. is expressible by the probability density $\mu_X(U) = \int_U f_X(x)dx$\textemdash then by the rule for integration under a diffeomorphic change of coordinates $y=g(x)$, 
\begin{align}
    \label{eq:change_coordinates_integral}
    \mu_X(U)(g^{-1}(U)) &= \int_{g^{-1}(U)}f_X(x)dx = \int_U f_X(g^{-1}(y))\left|\det\frac{\partial g^{-1}}{\partial y}(y)\right|dy.
\end{align}
with $\partial g^{-1}/\partial y$ denoting the Jacobian of the inverse. As this is true for all $U$ we conclude 
\begin{align}
    \label{eq:rv_transform_2}
    f_Y(y) &= f_X(g^{-1}(y))\left|\det\frac{\partial g^{-1}}{\partial y}(y)\right|.
\end{align}
When preferable to work with $g$ rather than $g^{-1}$, we may apply the chain rule to $g^{-1}\circ g = 1$ to convert \eqref{eq:rv_transform_2} to
\begin{align}
    \label{eq:rv_transform}
    f_Y(g(x)) &= f_X(x) \left|\det\left(\frac{\partial g}{\partial x}(x)\right)^{-1}\right| =: f_X(x)J^{-1}(x).
\end{align}
Due to this fact, we will abbreviate the scaling factor as $J^{-1}$, denoting the inverse Jacobian determinant.
For affine transforms $Y = cX + b$, \eqref{eq:rv_transform_2} becomes
\begin{equation}
    f_Y(y) = \frac1cf_X\left(\frac{y-b}{c}\right).
\end{equation}

\paragraph{Sums of random variables.} Let $X, Y$ be independent random variables taking values in $U_X, U_Y \subseteq \R$. Then their sum is distributed as the convolution
\begin{align}
    \label{eq:conv_rule}
    X + Y = Z &\sim f_Z(z) = \int_{x \in U_x,\ x \le z} f_X(x) f_Y(z-x)dx = f_X * f_Y.
\end{align}
This can alternatively be deduced by applying the transform rule \eqref{eq:rv_transform_2} to the map $(X, Y) \mapsto (X, X + Y)$. 

\begin{definition}[Convergence in distribution]
A sequence of random variables $\{X_n\}$ taking values in an interval $U \subseteq \R$ is said to {\it converge in distribution} if, for all $x$ at which the cumulative distribution function $F_X$ is continuous, the c.d.f.s. $F_{X_n}$ converge pointwise, i.e.
\begin{align}
    \lim_{n\to\infty}F_{X_n}(x) &= F_X(x).
\end{align}
\end{definition}

\subsection{Several properties of the exponential and gamma random variables}
\label{sec:rv_exp_gamma}

Let us recall with proof some properties of the gamma random variable (and the exponential, which is a special case). 

\begin{definition}[Gamma random variable.]
The gamma random variable $Y_{k,\lambda}$ with shape parameter $k > 0$ and rate $\lambda > 0$ is the continuous random variable with probability density function 
\begin{align}
    \label{eq:gamma_density}
    f_Y(y) &= \frac{\lambda^ky^{k-1}\exp(-\lambda y)}{\Gamma(k)}.
\end{align}
Its namesake is the normalizing constant, the gamma function 
\begin{align}
    \label{eq:gamma_fn}
    \Gamma(k) &= \int_0^\infty y^{k-1}e^{-y}dy.
\end{align}

\end{definition}

\begin{lemma}[Gamma random variables are closed under addition]
\label{lem:gam_sum}
The sum of two independent gamma random variables $Y_1, Y_2$ with $k_1,k_2 \in \R_{\ge 0}$ of common rate $\lambda > 0$ is $(k_1+k_2)$-gamma distributed.
\end{lemma}
\begin{proof}
Since $Y_1, Y_2 \in [0, \infty)$, by the convolution rule \eqref{eq:conv_rule},
\begin{align}
    \label{eq:sum_gamma}
    Y_{1} + Y_{2} &\pdfsim f_{Y_1} * f_{Y_2} =  \frac{\lambda^{k_1+k_2}}{\Gamma(k_1)\Gamma(k_2)} \int_0^{y_2} y_1^{k_1-1}e^{-\lambda y_1}(y_2-y_1)^{k_2-1} e^{-\lambda(y_2-y_1)}dy_1.
\end{align}
Taking the Beta function
\begin{align}
    \label{eq:beta_fn}
    \mathrm{B}(k_1,k_2) &= \int_0^1t^{k_1-1}(1-t)^{k_2-1}dt = \frac{\Gamma(k_1)\Gamma(k_2)}{\Gamma(k_1+k_2)} 
\end{align}
with the change of variable $t=y_1/y_2$ yields
\begin{align}
    &= \frac{1}{y_2^{k_1+k_2-1}} \int_0^{y_2}y_1^{k_1-1}(y_2-y_1)^{k_2-1}dy_1.
\end{align}
Substituting into \eqref{eq:sum_gamma} and letting $x=y_2$, we obtain
\begin{align}
    \eqref{eq:sum_gamma}&= \frac{\lambda^{k_1+k_2}x^{k_1+k_2-1}e^{-\lambda x}}{\Gamma(k_1+k_2)}\ \pdfsim\ Y_{;k_1+k_2,\lambda}.
\end{align}
\end{proof}

\begin{corollary}
    The sum of $k$ iid exponential random variables of rate $\lambda$ are gamma-distributed with shape parameter $k$ and rate $\lambda$.
\end{corollary}

\begin{corollary}
    The gamma random variable is infinitely divisible.
\end{corollary}

\begin{corollary}
    By the central limit theorem, $(\lambda Y_{k,\lambda} - 1) / \sqrt{k} \dconv W_{0,1}$ as $k \to \infty$.
\end{corollary}

\begin{lemma}[Beta-gamma convergence]
\label{lem:beta_gamma}
Let $Z_{2, m}$ be a beta random variable. Then $\frac{m}{\alpha}Z_{2,m}$ converges in distribution as $m \to \infty$ to the gamma random variable $Y_{2,\alpha}$ of rate $\alpha$. 
\end{lemma}
\begin{proof}
By direct calculation,
\begin{align}
    \frac{m}{\alpha}Z_{2,m} \cdfsim \mathrm{P}\left[Z_{2,m} \le \frac{z\alpha}{m}\right] &= \frac{1}{\mathrm{B}(2,m)}\int_0^{\frac{z\alpha}{m}}t(1-t)^{m-1}dt\qquad z\in\left[0,\frac{m}{\alpha}\right],
\end{align}
with the change of variables $s = 1-t$,
\begin{align}
    &= \frac{1}{\Beta(2,m)} \int_{1-\frac{\alpha z}{m}}^{1}(1-s)s^{m-1}ds\\
    &=m(m+1)\left[\frac{1}{m}s^m - \frac{1}{m+1}s^{m+1}\right]_{1-\frac{z\alpha}{m}}^1\\
    &= 1-\left((m+1)\left(1-\frac{z\alpha}{m}\right)^m - m\left(1-\frac{z\alpha}{m}\right)^{m+1}\right)\\
    &= 1-\left(1-\frac{z\alpha}{m}\right)^{m}(1+z\alpha)\\
    &\overset{m\to\infty}{\longrightarrow} 1-e^{-z\alpha}(1+z\alpha)\qquad z\in[0, \infty).
\end{align}
Recalling the cumulative density function for $Y_{2,\alpha}$ \eqref{eq:gamma_density},
\begin{align}
    Y_{2,\alpha} &\cdfsim \int_0^z \alpha^2ye^{-\alpha y}dy = -\alpha y e^{-\alpha y}\evalat_0^z + \int_0^z \alpha e^{-\alpha y}dy = 1-e^{-z\alpha}(1+z\alpha).
\end{align}
As the c.d.f.s of both are $C^\infty$ and exhibit the pointwise convergence above, we have the distributional convergence $\frac{m}{\alpha}Z_{2,m} \dconv Y_{2,\alpha}$.
\end{proof}

\begin{lemma}[Memoryless characterization of the exponential]
The only continuous random variable $X$ which (i) possesses a cumulative distribution function $F_X$ such that $F'_X(0)$ exists and (ii) satisfies the {\it memoryless} property
\begin{align}
    \label{eq:memoryless}
    \Prb[X > x + y\ |\ X > x] &= \Prb[X > y]
\end{align}
is the exponential.
\end{lemma}
\begin{proof}
By Bayes' theorem, \eqref{eq:memoryless} is equivalent to
\begin{align}
    \label{eq:memoryless2}
    \Prb[X > x+y] &= \Prb[X>x] \Prb[X>y]\ \Leftrightarrow\ 1-F(x+y) = (1-F(x))(1-F(y)).
\end{align}
Let $G = 1-F$; then \eqref{eq:memoryless2} is $G(x+y)=G(x)G(y)$. Then for all $z$,
\begin{align}
    \label{eq:memoryless5}
    G'(z) &= \lim_{h\to0}\frac{G(z)G(h)-G(z)G(0)}{h} = G(z)G'(0).
\end{align}
Since by hypothesis (i) $F'(0)$ exists, \eqref{eq:memoryless5} implies $F'$ exists everywhere.
Let $u$ be such that $F(u) < 1$. Since $F$ is nondecreasing, this implies $G > 0$ for $w \in (-\infty, u]$. Then, letting $G'(0) = -F'(0) = c$, 
\begin{align}
    c &= \frac{G'(w)}{G(w)} = \frac{d}{dw}\log G(w).
\end{align}
Integrating, we obtain $G(x) = b\exp(cx)$ for some $b$. By the conditions $F(0)=0, \lim_{t\to\infty}F(x) = 1$, we have $b=1, c < 0$. Letting $c = -\lambda$ for $\lambda > 0$  we obtain  $F_X(x) = 1-\exp(-\lambda x)$. Hence, $f_X(x) = F_X'(x) = \lambda\exp(-\lambda x)$, and  $X$ is an exponential random variable of rate $\lambda$. 
\end{proof}

\begin{lemma}[Maximum-entropy characterization of the exponential]
\label{lem:exp_maxent}
The only nonnegative continuous random variable $X$ with density $f_X$ which maximizes the entropy with fixed mean $\mu > 0$ is the exponential.  
\end{lemma}
\begin{proof}
By hypothesis, $f_X$ is a critical point of the functional
\begin{align}
    J[f] &= \int_0^\infty L(f(x), \lambda_0, \lambda_1)dx,
\end{align}
with Lagrange multipliers $\lambda_0$ and $\lambda_1$ constraining the 0th and 1st moments in the Lagrangian density
\begin{align}
    L(f(x), \lambda_0, \lambda_1) &= f(x)\log f(x) + \lambda_0 f(x) + \lambda_1 xf(x).
\end{align}
Since for all test functions $\varphi$, the Fr\'echet derivative vanishes,
\begin{align}
    0 &= \langle DJ[f], \varphi\rangle  = \int_0^\infty \frac{\partial L}{\partial f(x)} \varphi(x) dx,
\end{align}
by the fundamental lemma of the calculus of variations,
\begin{align}
    0 &= \frac{\partial L}{\partial f(x)} = 1 + \log f(x) + \lambda_0 + x\lambda_1,
\end{align}
hence $f(x) = \exp(-1 - \lambda_0 - x\lambda_1)$. Applying the total probability constraint,
\begin{align}
    1 &= \int_0^\infty \exp(-1-\lambda_0-x\lambda_1)dx = \frac{1}{\lambda_1} \exp(-1-\lambda_0),
\end{align}
where $\lambda_1 > 0$ necessarily. Then, $f_X(x) = \lambda_1\exp(-\lambda_1 x)$ is the density of an exponential random variable with rate $\lambda_1$.
\end{proof}

\begin{theorem}[Characterization of gamma random variables, Lukacs 1955 \cite{LukacsIndependence}]
\label{thm:gamma_independence}
    Let $Y_1, Y_2$ be independent random variables. Then 
    \begin{align}
        A &= Y_1 + Y_2,\qquad B = \frac{Y_1}{Y_1 + Y_2}
    \end{align}
    are independent if and only if $Y_1, Y_2$ are gamma random variables of the same rate $\lambda$.
\end{theorem}
\begin{proof}[Proof ($\implies$)]
Consider  the map
\begin{align}
    g(Y_1, Y_2) &= \left(Y_1+Y_2, \frac{Y_1}{Y_1+Y_2}\right) =: (A, B).
\end{align}
Then for $a\ne 0$,
\begin{align}
g^{-1}(a,b) &= (ab, a-ab).
\end{align}
Therefore
\begin{align}
    \left|\det\frac{\partial g^{-1}}{\partial (a,b)}\right| &= \left|\det\begin{bmatrix}b & 1-b\\ a & -a\end{bmatrix}\right| = a.
\end{align}
By \eqref{eq:rv_transform}, the pushforward density is
\begin{align}
    \label{eq:lukacs_joint_density}
    f_{g(Y_1,Y_2)}(a, b) &= af_{Y_1}(ab)f_{Y_2}(a-ab) \\
    &= \frac{a(ab)^{k_1-1}(a-ab)^{k_2-1}}{\Gamma(k_1)\Gamma(k_2)}\lambda^{k_1+k_2}e^{-\lambda a}\\
    &= \frac{b^{k_1-1}(1-b)^{k_2-1}}{\Gamma(k_1)\Gamma(k_2)}a^{k_1+k_2-1}\lambda^{k_1+k_2}e^{-\lambda a}.
\end{align}
The total $A$ and fraction $B$ are therefore independent.
Substituting the Beta function \eqref{eq:beta_fn},
\begin{align}
    \eqref{eq:lukacs_joint_density}&= \frac{b^{k_1-1}(1-b)^{k_2-1}}{\mathrm{B}(k_1, k_2)}\frac{a^{k_1+k_2-1}\lambda^{k_1+k_2}e^{-\lambda a}}{\Gamma(k_1+k_2)} =: f_B(b)f_A(a).
\end{align}
Then $f_A(a)$ is a gamma distribution (as expected, \ref{lem:gam_sum}) and $f_B(b)$ is a beta distribution. Therefore, 
\begin{align}
    \label{eq:beta_gamma}
    B = \frac{Y_{1;k_1,\lambda}}{Y_{1;k_1,\lambda} + Y_{2;k_2,\lambda}} &= Z_{k_1,k_2}
\end{align}
is beta-distributed for $k_1, k_2 \in \R_{\ge 0}$.
\end{proof}
Note that \eqref{eq:beta_gamma} is independent of the rate $\lambda$ of $Y_1, Y_2$. The converse $(\impliedby)$, that gamma distributions are unique in possessing this independence property, is not proven here, but we refer the reader to a proof \cite{FindeisenGamma} using the fact that the gamma distribution is uniquely determined by its moments (3.3.25, \cite{ProbabilityTheory}).

\begin{corollary}[Beta-thinned gamma random variable]
    If $Y_{1;k_1+k_2,\lambda}$ and $Z_{1;k_1,k_2}$ are independent gamma and beta random variables, respectively, then
\begin{align}
    Y_{2;k_1,\lambda} &= Z_{1;k_1,k_2}Y_{1;k_1+k_2,\lambda}
\end{align}
is an independent gamma random variable of the same rate $\lambda$ and lower shape parameter $k_1$.
\end{corollary}
The {\it thinning} here refers to the fact that $Z_1$ is supported on $[0, 1]$, hence decreasing the number of events occurring in a section of a 1D Poisson process.

\begin{lemma}[Differential entropy of a fixed-mean gamma random variable is strictly decreasing in $k \in (1, \infty)$]
\end{lemma}
\begin{proof}
Let $Y_{k, \lambda}$ be a gamma random variable of shape parameter $k$ and rate $\lambda$. Its differential entropy is
\begin{align}
    H(Y) &= k + \log \Gamma(k) + (1-k)\psi(k) - \log \lambda, 
\end{align}
where $\psi(k) = \frac{d}{dk}\log\Gamma(k)$ is the digamma function.
Hence, at fixed mean $\mu = k / \lambda$,  using the hypothesis $k > 1$ and the trigamma inequality $\psi'(k) > 1/k$ \cite{ProbabilityTheory}, we have the $k$-monotonicity
\begin{align}
    \frac{\partial H}{\partial k} &= 1 + \psi'(k) (1-k) - \frac{1}{k} < 1 + \frac{1}{k}(1-k) - \frac{1}{k} = 0.
\end{align}
\end{proof}
Note furthermore that $H(Y) \ge 0$ for $k \le 1$, so one recovers the maximum-entropy property of the exponential random variable \ref{lem:exp_maxent}.

\textbf{A note on maximum likelihood estimation.}
Throughout, we use standard methods for maximum-likelihood estimation of the shape and scale parameters $(k, \lambda)$ of the gamma distribution at fixed offset $0$, as implemented in Python's \textit{scipy.stats.gamma.fit} function.
As the maximum likelihood estimate of $\lambda$ is given by the $k / \overline{x}$ where $k$ is given and $\overline{x}$ is the empirical mean, we estimate 95\% confidence intervals for the shape parameter $k$ at fixed scale $\lambda$ using $N=1000$ parametric bootstraps as follows. 
Each sample of size $n$, where $n$ is the size of the original dataset to be fit, is produced from a gamma distribution with $(\hat{k}, \hat{\lambda})$, the estimated parameters from the original empirical distribution.
Then, at fixed $\hat{\lambda}$, a new $\Tilde{k}$ is maximum-likelihood estimated from these samples. 
From this set $\{\Tilde{k}_i\}_{i=1}^N$, the values corresponding to the $(2.5, 97.5)$ percentiles are reported as the 95\% CIs.

\section{1-dimensional models of cells within an ECM}

Recall that the configuration of a sequence of cellular positions $\{R_i\}$ in a one-dimensional ECM (e.g. $[0, \infty)$ or the circle $\SP^1$) is uniquely defined, up to relabeling, by the intercellular spacings $L_i = R_{i+1} - R_i$. 
In the following models, we let $L_i$ be given (up to appropriate scaling) by the steady-state protein concentration $P_i^*$ referred to in the main text.
Realizing random 1D configurations in different ways by applying constraints analogous to the ensembles of statistical physics, we obtain the same $k=2b$-gamma distribution governing the Voronoi lengths in a regime analogous to the thermodynamic limit.
These models are named by the ECM dimension $d$ (I or II) and described by the particular ensemble considered.
as summarized in Table \ref{tab:pp_models}.

\begin{table}[h]
\caption{Models in 1D and 2D}
\centering
\label{tab:pp_models}
\begin{tabular}{l l}
\toprule
Model & Description \\ 
\midrule
\hyperlink{I.0}{I.0} & on the half-line  \\
\hyperlink{I.1}{I.1} & circular Gibbs \\
\hyperlink{I.2}{I.2} & circular canonical \\
\hyperlink{I.3}{I.3} & circular grand canonical \\
\hyperlink{I.4}{I.4} & circular canonical, finite size \\
\hyperlink{I.D.1}{I.D.1} & circular canonical, Brownian motion \\ 
\hyperlink{I.D.2}{I.D.2} & circular canonical, noncolliding Brownian motion\\ 
\hyperlink{I.D.3}{I.D.3} & noncolliding Brownian motion on a growing circle\\ 
\hyperlink{I.D.4}{I.D.4} & maximum-entropy growth rates\\
\hyperlink{II.1}{II.1} & canonical on periodic unit square\\
\bottomrule
\end{tabular}
\end{table}

{\it Notation.} We denote random variables by capital letters $X_{i;\alpha,\beta,\ldots}$ accompanied by indices $i$ and parameters $\alpha,\beta,\cdots$.
Variables $i \ne j$ are independent unless otherwise noted. 
The letters $W_{i;\mu,\sigma^2}, X_{i;\lambda}$, $Y_{i;k,\lambda}, Z_{i;\alpha,\beta},$ and $N_{i;\lambda}$, and are reserved for Gaussian, exponential, gamma, beta, and Poisson random variables respectively. 
$X \pdfsim f_X(x)$ indicates that $X$ has the probability density function $f_X(x)$, with $X \cdfsim F_X(x)$ indicating the same for the cumulative density. 
$X_i \dconv Y$ indicates that $X_i$ converges to $Y$ in distribution.
For example, we would say that $Y_{k,\frac{k}{\overline{v}-v_c}} + v_c \pdfsim p(v)$, where $p(v)$ is the distribution (1).

As in the Gibbs, microcanonical, and grand canonical ensembles, we consider three types of random configurations of cells in the ECM: (i) fixed cell counts $N=n$, with random intercell spacings and circumferences; (ii) fixed circumferences $C=c$ and counts $N=n$, with random positions $\{R_i\}_{i=1}^N$; and (iii) fixed circumferences $C=c$, with random cell counts $N$ and positions. 
Like the convergence of the ensembles in the thermodynamic limit, we show that the same gamma distribution arises in the large-$n,c$ limits of these cases.

\hypertarget{I.0}{}
\textit{Model I.0 (on the half-line)} 
\textemdash Consider a semi-infinite {\it Volvox} modelled as a sequence of cellular positions $\{R_i\}_{i=1}^\infty$ on the half-line $[0, \infty)$ as in Fig. 1(c).
In the special case of $b=1$ protein bursts, $R_i$ are the cumulative sums of $i$ i.i.d. exponential random variables:
\begin{align}
    \label{eq:1d_poisson}
    R_i &= \sum_{j=1}^i X_{j,i;\lambda},
\end{align}
and therefore occur on the half-line $[0, \infty)$ as a Poisson process \cite{DaleyPointProcess}.
This is the configuration of the ideal gas\textemdash or the maximum-entropy configuration\textemdash in which the number of cells observed in any interval of length $\ell$ is a Poisson random variable $N_{\lambda \ell}$ whose positions are i.i.d. uniform random variables \cite{DaleyPointProcess}. 
This is the case of class (i) configurations.

The Voronoi lengths $V_i = (L_i + L_{i+1})/2 = (X_{i;\lambda} + X_{i+1;\lambda})/2$ are therefore gamma-distributed with $k=2$,
\begin{equation}
    V_i = \frac{X_{i;\lambda}+X_{i+1;\lambda}}{2} =  \frac12Y_{i;2,\lambda} \pdfsim 4\lambda^2ve^{-2\lambda v}.
    \label{eq:1d_reals_gamma}
\end{equation}
This result already shows explicitly the deep link between the Voronoi construction and gamma distributions. 
More generally, for burst count $b$, the resulting Voronoi lengths, by the same argument, are $k=2b$-gamma random variables.

In the opposite limit, holding the mean spacing $\E[L_i] = b  / \lambda$ fixed while taking the burst count $b \to \infty$, the variance $\sigma^2 = b/\lambda^2 = \E[L_i]^2 / b$ vanishes while the central limit theorem ensures the convergence of the shifted and rescaled lengths $\sqrt{b}\left(\frac{L_i}{\E[L_i]} - 1\right) \to W_{i;0,1}$ to a Gaussian.
This is the perfectly spaced lattice of cellular positions occurring as the natural numbers $\N$ on $[0, \infty)$\textemdash the ``crystalline,'' or class (i) configuration.

\hypertarget{I.1}{}
\textit{Model I.1 (circular Gibbs)} \textemdash 
Consider a circular {\it Volvox}, constructed by selecting a fixed number $N_{;\lambda}=n+1$ of successive points $\{R_i\}_{i=0}^{n+1}$ from the half-line in \hyperlink{I.0}{I.0} and identifying the first and last points, as in Fig. 1(c). 
This circle has a random circumference $C$ (``Gibbs ensemble'') which is $nb$-gamma distributed, with resulting Voronoi lengths $V_i = Y_{i;2b,2\lambda}$ governed by $k=2b$-gamma distributions by an identical argument as in \hyperlink{I.0}{I.0}.
Fig. 1(d) shows an empirical distribution of $V_i$ from $10^4$ samples.

\hypertarget{I.2}{}
\textit{Model I.2 (circular canonical)}\textemdash 
Consider a fixed-$N$, fixed-$C$ configuration as follows.
Let us take the intercellular spacings $L_i$ to be $b$-gamma random variables $Y_{i;b,\lambda}$ as in \hyperlink{I.1}{I.1}, \textit{conditional} on their sum $\sum_i Y_i = C$. 
That is, if $A = Y_{1;b,\lambda}$ is one spacing and $B=\sum_{j=2}^nY_{j;b,\lambda} = Y_{B;(n-1)b,\lambda}$ is the rest, then the distribution of $A$ given $A+B=C$ is
\begin{align}
    \label{eq:conditional_sum_density}
    f_{A|C}(a|c) &= \frac{a^{b-1}(c-a)^{(n-1)b-1}}{\Beta(b,(n-1)b)c^{nb-1}}
\end{align}
which one shows by the fraction-sum independence property of gamma random variables \ref{thm:gamma_independence}, where $\Beta(\alpha,\beta)$ is the Beta function.
Then, by inspection of \eqref{eq:conditional_sum_density} and scaling laws for r.v.s, $L_i$ is the beta random variable $A|C = cZ_{i;b,b(n-1)}$ with $C=c$ now the fixed circumference.

Once more, in the special case of protein burst count $b=1$, $L_i = cZ_{i;1,n-1}$, which is simply the first order statistic of $(n+1)$ i.i.d. uniform random variables on the interval $[0, c]$ (with $n+1$ arising from identifying the ends of the interval to form a circle).
Then, this is the distribution of \textit{uniform spacings} of the interval $[0, c]$, which is precisely the distribution of waiting times for $n$ events in a Poisson process conditional on a total wait time $C=c$ (see \S 4.1, \cite{PykeSpacings}), as shown in Fig. 1(c). 

Since $Z_{i;b,b(n-1)}$ can be expressed as the fraction $Y_{i;b,\lambda}/\sum_{j=1}^nY_{j;b,\lambda}$ \ref{thm:gamma_independence}, this allows us to compute the corresponding Voronoi lengths $V_i$ as
\begin{align}
    \label{eq:1d_circle_beta}
    V_i &= \frac{c}{2}\frac{Y_i+Y_{i+1}}{\sum_{j=1}^nY_{j}} = \frac{c}{2}Z_{j;2b,(n-2)b},
\end{align}
with subscript $Z_j$ indicating that $Z_j$ is a distinct (though not independent) random variable from the earlier $Z_i$.
In the Poisson case $b=1$, this is now simply one-half the \textit{second} order statistic of i.i.d. uniform r.v.s. on  $[0, c]$.

Taking the large-cell count and ECM circumference limits $n, c\to\infty$ limit at fixed cell number density $\rho=m/c=(n-2)/c$ (analogous to the thermodynamic limit of statistical mechanics) and per-cell burst count $b$, one obtains \ref{lem:beta_gamma} the convergence in distribution of the Voronoi segments,
\begin{eqnarray}
    \label{eq:circle_gamma_convergence}
    V_i &=& \frac{mb}{2\rho}Z_{;2b,mb}  \overset{m\to\infty}{\longrightarrow} \frac12 Y_{;2b,\rho},
\end{eqnarray}
whose limit is once again the $k=2b$ gamma random variable (with rate $\rho$) in \hyperlink{I.0}{I.0} and the Gibbs ensemble \hyperlink{I.1}{I.1}. 

\hypertarget{I.3}{}
\textit{Model I.3 (circular grand canonical)}\textemdash 
Let the circumference $C=c$ now be fixed with the count $N$ a random variable. 
In the case of $b=1$ bursts, $N = N_{;c\lambda}$ is Poisson-distributed.
Using a similar argument to \hyperlink{I.2}{I.2}, the Voronoi length distribution $p_{V_i}$ can be determined by marginalizing the joint distribution $p_{Z_{;2,n-1},N_{;\lambda}}$ over $n$, giving the {\it compound beta-Poisson distribution} which is once again a $k=2$ gamma distribution, shown as follows.

As in \hyperlink{I.2}{I.2}, let us consider the $k$th order statistic of $N$ uniform random variables (representing the fixed-circumference constraint), with $N_\lambda$ now a Poisson-distributed random variable conditioned to be minimum $k$. The marginal distribution of the order statistic given $\lambda$ is (a particular) {\it compound beta-Poisson distribution}, given by
\begin{eqnarray}
    \label{eq:compound_order_statistic}
    \Prb(Z_{k, N-k+1}=x|\lambda, N\ge k) &=& \sum_{n=k}^\infty \Prb(Z_{k,n-k+1}=x|N=n)\Prb(N=n|\lambda,N \ge k).
\end{eqnarray}
To see the parametrization more clearly, let $m = n - k + 1$; then, 
\begin{eqnarray}
    \eqref{eq:compound_order_statistic} &=&\frac{1}{1-\Prb(N < k)} \sum_{m=1}^\infty\frac{x^{k-1}(1-x)^{m-1}}{\Beta(k,m)}\frac{\lambda^{m-1+k}e^{-\lambda}}{\Gamma(m+k)}\\
    &=& \frac{1}{1-\Prb(N < k)}\frac{\lambda^k x^{k-1}e^{-\lambda}}{\Gamma(k)}\sum_{m=1}^\infty \frac{(\lambda(1-x))^{m-1}}{\Gamma(m)}\\
    &=& \frac{1}{1-\Prb(N < k)}\frac{\lambda^k x^{k-1}e^{-\lambda x}}{\Gamma(k)}\\
    &=:& \frac{1}{Z} \frac{\lambda^k x^{k-1}e^{-\lambda x}}{\Gamma(k)},\qquad x \in [0, 1].
\end{eqnarray}
Recognizing the second factor as the gamma distribution, the first factor $Z$ is simply a normalizing constant restricting the support to $[0, 1]$. 
Here we see the emergence of gamma distributions from the order statistics of a Poisson-distributed number of uniform random variables. 
This may be viewed as a roundabout method of constructing the 1D Poisson process. 
As we take the support of this distribution (the circumference $C$) to infinity, we have $P(N_{C\lambda} < k) \to 0$ for any fixed $k$, hence $Z \to 1$ and we recover the true gamma distribution.

For the general $b$ case, we note that this construction is equivalent to placing an observation window $[x, x+c]$ at random on the half-line process \hyperlink{I.0}{I.0}, in which case the spacings $L_i$ follow a truncated $b$-gamma distribution:
\begin{eqnarray}
    \label{eq:truncated_exponential}
    L_i &\pdfsim& \frac{\lambda^b x^{b-1}e^{-\lambda x}}{\gamma(b, \lambda c)},\qquad x \in [0, c]
\end{eqnarray}
with $\gamma$ the lower incomplete gamma function.
Taking $c\to\infty$ with rate $\lambda$ fixed (thermodynamic limit), we have $\gamma(b, \lambda c) \to \Gamma(b)$, hence $L_i \dconv Y_{i;b,\lambda}$ and $V_j \dconv Y_{j;2b,2\lambda}$, giving the same $k=2b$-gamma distribution for Voronoi lengths. 






\begin{figure*}[t]
\includegraphics[width=\columnwidth]{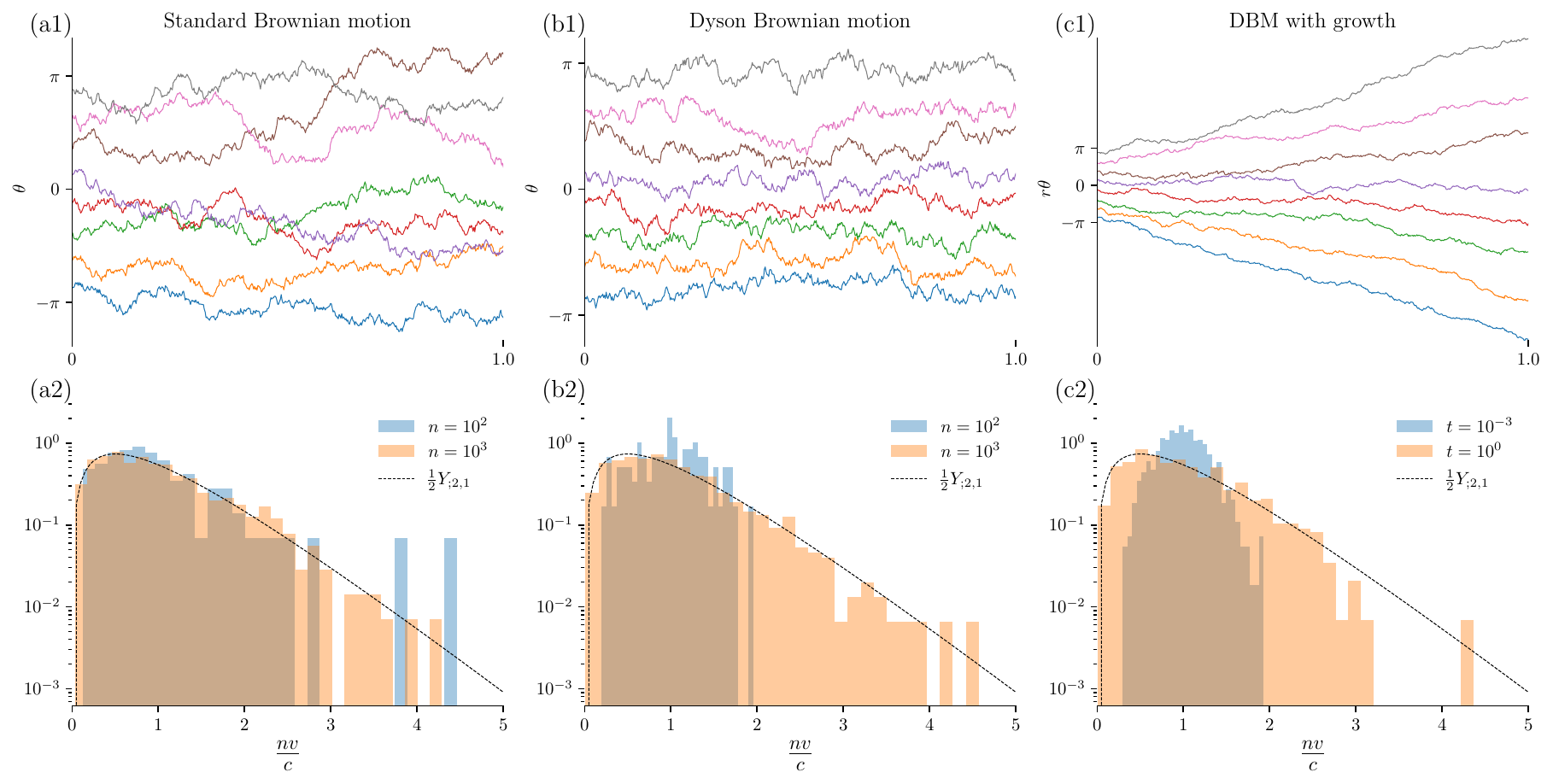}
\caption{Dynamic models of cellular positions during growth converge to gamma distributions in various limits. (a1,a2) Standard Brownian motion in angular coordinates converges to beta-distributed segments in large $t$, and gamma-distributed segments in large $n$. (b1,b2) Dyson Brownian motion in angular coordinates converges to gamma-distributed segments at large $t$ and $n$. Non-conformance to gamma distributions is observed at low $n$ due to pair-repulsion. (c1,c2) Dyson Brownian motion on a growing domain satisfying particular growth constraints converges in large-$t$ and $n$ to gamma-distributed segments.}
\label{fig:brownian}
\end{figure*}


\hypertarget{I.4}{}
\textit{Model I.4 (circular canonical, finite size)}\textemdash 
This refers to the fixed-cell size model discussed and analyzed in the main text.

\hypertarget{I.D.1}{}
\textit{Model I.D.1 (Brownian motion on the circle)}\textemdash Let $\{R_t^{(i)}\}_{i=1}^n$ be $n$ Brownian motions (BMs) on a circle of radius $r$ given by 
\begin{eqnarray}
    \label{eq:circular_motion}
    R_t^{(j)} &=& r\exp(i\theta_t^{(j)}), 
\end{eqnarray}
with $\theta_t^{(j)}$ standard BMs on $\mathbb{R}$ and initial conditions for the probability density $p^{(j)}(0) = \delta_{R_0^{(j)}}$. 
The time-dependent probability density $p^{(i)}(t)$ tends exponentially in $L^2$ to the uniform distribution, hence $R_t^{(i)}$ converge in distribution to i.i.d. uniform r.v.s. and the configuration approaches that of Model \hyperlink{I.2}{I.2}.
$V_i$ are therefore beta-distributed \eqref{eq:1d_circle_beta} in large $t$ and gamma-distributed in large $t,n$ \eqref{eq:circle_gamma_convergence}\textemdash see Fig. \ref{fig:brownian}(a1,a2). 
Yet, the sample paths of \eqref{eq:circular_motion} almost surely intersect, unphysically reordering the cells.

\hypertarget{I.D.2}{}
\textit{Model I.D.2 (noncolliding Brownian motion)}\textemdash Let $\{R_t^{(i)}\}_{i=1}^n$ be samples of the conditional distributions of $n$ circular Brownian motions whose angles are 
in ascending order modulo $2\pi$, thereby lying in the set
\begin{eqnarray}
    \label{eq:weyl_chamber}
    D_n &=& \{x\in\mathbb{R}^n\ |\ x_1 < \cdots < x_n < x_1 + 2\pi\},
\end{eqnarray}
a construction known as {\it noncolliding Brownian motion} 
or Brownian motion within the {\it Weyl chamber} $D_n$ 
\cite{GrabinerBrownian}. In \cite{HobsonNoncolliding} 
eq. 4.1, it is shown that the conditional fluctuations 
are Gaussian plus a singular $r^{-1}$ pair-repulsion,
\begin{eqnarray}
    \label{eq:dyson_bm}
    d\theta_t^{(i)} &=& \sigma dB_t^{(i)} + \frac{\sigma^2}{2}\sum_{j\ne i} \cot\left(\frac{\theta_t^{(i)} - \theta_t^{(j)}}{2}\right)dt.
\end{eqnarray}
A physical interpretation of \eqref{eq:dyson_bm} is 
that of a gas confined to the unit circle with the 
pair-potential 
\begin{eqnarray}
\label{eq:dyson_potential}
    U &=& -\sum_{j < k}\log\left|\exp(i\theta_k) - \exp(i\theta_j)\right|,
\end{eqnarray}
constituting a simple model of non-colliding cell motion. 
Eq. \ref{eq:dyson_bm} is precisely the eigenvalue dynamics $\lambda_t^{(j)} = \exp(i\theta_t^{(j)})$ of a Brownian 
motion $U_t$ on the unitary group $\mathbb{U}(n)$, known as \textit{Dyson Brownian motion} \cite{DysonBrownian}.

Being confined to $D_n$ \eqref{eq:weyl_chamber}, 
the positions $R_t^{(i)}$ do not converge to uniform 
r.v.s on the circle as in I.5\textemdash compare Figs.
\ref{fig:brownian}(a1) and \ref{fig:brownian}(b1). The
stationary distribution of \eqref{eq:dyson_bm} is the
circular unitary ensemble (CUE) \cite{DysonBrownian},
\begin{equation}
    \label{eq:circular_ensemble}
    \rho_*(\theta_1,\cdots,\theta_n) = \frac{1}{Z_{n}}\prod_{j<k}|\exp(i\theta_j) - \exp(i\theta_k)|^\beta,
\end{equation}
where $\beta = 2$ is the inverse temperature.
We use this result to derive the Voronoi length distribution in the large-$t,n$ limits. 
Let $\{\theta^{(i)}\}_{i=1}^n$ be a sample of the stationary distribution \eqref{eq:circular_ensemble}, the spectrum of a uniform sample of  $\mathbb{U}(n)$; the empirical distribution $\mu_\theta(n) = n^{-1}\sum_{i=1}^n\delta_{\theta^{(i)}}$ converges almost surely in large $n$ (Theorem 3, \cite{DiaconisEigenvalues}) to the uniform distribution on the unit circle. 
Thus, the spacings $\theta^{(i+1)}_t-\theta^{(i)}_t$ converge in large-$t$ and $n$ to the spacings of the order statistics of $n$ uniform random variables as in \hyperlink{IS2}{IS.2}. 
Taking $n\to\infty$ with density $\rho=(n-2)/c$ fixed, $V_i$ converge to $k=2$ gamma random variables as in \eqref{eq:circle_gamma_convergence}.

Convergence to gamma laws depending on particle count $n$ is shown for BM and DBM in Figs. \ref{fig:brownian}(a2,b2); in contrast to BM, DBM spacings lose the long tail at low $n$ due to repulsion \eqref{eq:circular_ensemble}. 

\hypertarget{I.D.3}{}
\textit{Model I.D.3 (noncolliding Brownian motion, growth)}\textemdash 
To account for growth of the ECM, 
let the radius $r(t)$ in \hyperlink{A.I.1}{A.I.1} be time-dependent, 
scaling the configuration as $R_t^{(i)} = r(t)\exp(i\theta_t^{(i)})$ 
with $\theta_t^{(i)}$ given by \eqref{eq:dyson_bm} as shown in Fig. 
\ref{fig:brownian}(c1). Applying It\^o's lemma, $R_t^{(i)}$ satisfies 
the stochastic differential equation
\begin{equation}
    dR_t^{(i)} = \dot{r}\exp(i\theta_t^{(i)})dt + iR_t^{(i)}d\theta_t^{(i)} - \frac{r}{2}\theta_t^{(i)}dt.
    \label{eq:growth_sde}
\end{equation}
Substituting $d\theta_t^{(i)}$ from \eqref{eq:dyson_bm}, the diffusion constant for \eqref{eq:growth_sde} is
$D=r^2\sigma^2/2$.
Requiring that the lateral diffusion is time invariant implies $\dot{D} = 0$, and thus that 
the standard deviation of the radius-normalized dynamics \eqref{eq:dyson_bm} should decay as $\sigma(t) \propto r(t)^{-1}$.
Exponential convergence of \eqref{eq:dyson_bm} to the stationary solution \eqref{eq:circular_ensemble} is ensured by a growth condition on $r(t)$. 
By standard Fourier arguments (e.g. \S 2.2, \cite{KempHeat}), solutions to the time-dependent diffusion equation for the probability density of $p(t)$ of the unitary Brownian motion $U_t$ satisfy, for some Poincar\'e-like constant $C_n > 0$ depending only on $n$,
\begin{equation}
    d(p(t), \nu) \le d(p(0), \nu) \exp\left(-C_n\int_0^t\sigma(s)^2ds\right),
\label{eq:unitary_bm_convergence}
\end{equation}
with $d$ the $L^2$ metric and $\nu$ the uniform (Haar) measure on the unitary group $\mathbb{U}(n)$. Therefore, if $r(t)$ satisfies:
\begin{eqnarray}
    \label{eq:growth_condition}
    \lim_{t\to\infty}\int_0^t \frac{1}{r(s)^2}ds &=& \infty, \qquad r(0) > 0
\end{eqnarray}
the RHS in \eqref{eq:unitary_bm_convergence} vanishes in large $t$, and global exponential stability is assured. Then, as in \hyperlink{ID2}{ID.2}, the positions $R_t^{(i)}$ approach a random uniform spacing of a circle of radius $r(t)$ in large $n$, resulting in gamma-distributed Voronoi segments $V_i$. Rapid convergence in $t$ to a gamma law (with $D = 0.1, n=1000, \dot{r}=1$) is shown in Fig. \ref{fig:brownian}(c2). Formally, \eqref{eq:growth_condition} is not satisfied by such a linear growth law, yet we observe empirical convergence to a gamma law on the order of a unit of (scaled) time.

Condition \eqref{eq:growth_condition} can  be understood by a P\'eclet number relating drift and diffusion timescales in growth. Considering the radial drift velocity $v_r$ in \eqref{eq:growth_sde}, let
\begin{equation}
    \label{eq:growth_peclet}
    \text{Pe} = \frac{\tau_{d}}{\tau_{v_r}} = \frac{2L\dot{r}}{r^2\sigma^2} \propto L\dot{r}
\end{equation}
with $L$ a test length section.
In the limit $\text{Pe} \to 0$,  condition \eqref{eq:growth_condition} is  satisfied; when $\text{Pe} \to \infty$, the exponential multiplier in the bound \eqref{eq:unitary_bm_convergence} approaches $1$, ``freezing'' the angles $\theta^{(i)}_t$ to initial conditions. Finally, conditions $\dot{D} = 0$, $\dot{\text{Pe}} = 0$, and \eqref{eq:growth_condition} cannot simultaneously be satisfied. 

\hypertarget{I.D.4}{}
\textit{Model I.D.4 (maximum-entropy growth rates)}\textemdash Let the segments grow linearly in time as $\dot{L}_i = G_{i;\mu}$, with i.i.d. growth rates $G_i$ which are positive continuous random variables with some common mean growth rate $\mu$. If the distribution of $G_i$ maximizes entropy subject to the mean and nonnegativity constraint\textemdash perhaps more interpretable as uncertainty about cellular behavior than global maximum-entropy assumptions \cite{AsteEntropy}\textemdash then $G_{i;\mu} = X_{i;1/\mu}$ is an exponential random variable of rate $\lambda=1/\mu$. Then, for any time $t$, the normalized configuration 
\begin{align}
    \frac{L_i(t)}{\sum_{j=1}^nL_j(t)} &= \frac{L_i(0) + tG_i}{\sum_{j=1}^nL_j(0) + tG_j} \dconv \frac{X_{i;1/\mu}}{\sum_{j=1}^nX_{j;1/\mu}}
\end{align}
converges in large $t$ to a uniform spacing as in \hyperlink{I2}{I.2}, and therefore has beta-distributed Voronoi lengths converging to gamma-distributed lengths in large $n$ \eqref{eq:circle_gamma_convergence}.

\section{Defining point processes in $d$ dimensions}
Many of the following definitions are standard and recalled here for a self-contained reference.

\subsection{Definitions in  $\R^d$}
The familiar Poisson process of rate $\lambda$ on the half-line $[0, \infty)$ is constructible as the cumulative sum of i.i.d. exponential random variables $X_{i;\lambda}$. The construction of general point processes on a domain $K \subseteq \mathbb{R}^d$, however, is more technical, formalizing the notion of a ``random almost-surely finite subset,'' for which we recall several standard definitions \cite{DaleyPointProcess}. Throughout, we assume $K$ is a complete separable metric space (c.s.m.s.), e.g. a closed subset of $\R^n$.

\begin{definition}[Finite point process \textemdash 2.2-2.4, \cite{DaleyPointProcess}]
\label{def:finite_pp}
A finite point process $N$ on a complete separable metric space $K$ is a family of random variables $N(E)$ for each Borel set $E \in \mathcal{F}_K$, such that, for every bounded $E$,
\begin{align}
    \Prb[N(E) < \infty] &= 1.
\end{align}
\end{definition}

Formally, $N$ is a {\it random measure} (see e.g. 5.1 \cite{DaleyPointProcess} or Chapter 9 \cite{DaleyPointProcess}). Without delving too deeply into this formalism, let us introduce the following definition based upon samples of $N$. 

\begin{definition}[Simple point process]
\label{def:simple_pp}
A {\it simple} point process $N$ is one whose points are non-overlapping. In other words, every sample of $N$ can be written as the counting measure
\begin{align}
    \nu &= \sum_{i\in I} \delta_{x_i}
\end{align}
where $I$ is an index set, $\delta$ denotes the Dirac measure, and $\Prb[x_i = x_j] = 0$ for all $i\ne j$.
\end{definition}

\begin{definition}[Non-atomic point process]
\label{def:nonatomic_pp}
A {\it nonatomic} point process $N$ is one for whom the probability of realizing any particular point $x \in K$ is zero. That is,
\begin{align}
    \Prb[N(\{x\}) > 0] &= 0.
\end{align}
\end{definition}

These definitions are sufficient to define and analyze the Poisson point process.

\begin{definition}[Poisson point process]
\label{def:poisson_pp}
A Poisson point process $N$ on a c.s.m.s. $K$ is defined by an \textit{intensity measure} $\Lambda(E)$ such that, for all Borel sets $E \in \mathcal{F}_K$, $N(E)$ is a Poisson random variable, i.e.
\begin{align}
    \label{eq:poisson_count_dist}
    N(E) &\pdfsim \frac{\Lambda(E)^k\exp(-\Lambda(E))}{k!}.
\end{align}
\end{definition}

\begin{example}[Stationary Poisson process]
A {\it stationary} Poisson process of rate $\lambda$ is given by the intensity measure $\Lambda(E) = \lambda \mu(E)$, where $\mu$ is the Lebesgue measure and $\lambda > 0$ is a positive rate.
\end{example}

\begin{definition}[Independent scattering / complete independence]
\label{def:independent_pp}
A point process $N$ satisfies the {\it independent scattering} or {\it complete independence} property if, for all $n > 1$ and disjoint Borel sets $E_1,\ldots,E_n \in \mathcal{F}_K$, the variables $N(E_1),\ldots,N(E_n)$ are mutually independent.
\end{definition}

\subsection{Characterizing Poisson point processes}
\label{sec:poisson_characterization}

Immediately, we see that Poisson point processes are simple and finite if the intensity measure $\Lambda(E)$ is given by 
\begin{align}
    \Lambda(E) &= \int_E \lambda(x)dx
\end{align}
for some function $\lambda : K \to \R_+$. Poisson point processes \ref{def:poisson_pp} satisfy the complete independence property \ref{def:independent_pp}\textemdash but remarkably, these properties are not logically independent. 
 

\begin{theorem}[Prekopa 1957, Theorem 2.4.V \cite{DaleyPointProcess}]
\label{thm:poisson_pp_characterization}
A point process $N$ is a non-atomic Poisson point process if and only if it is finite \ref{def:finite_pp}, simple \ref{def:simple_pp}, non-atomic \ref{def:nonatomic_pp}, and completely independent \ref{def:independent_pp}.
\end{theorem}

The characterization theorem \ref{thm:poisson_pp_characterization} motivates and justifies the use of Poisson processes (i.e. Poisson-distributed count variables) in any scenario where points in realizations are non-interacting.


\subsection{Sampling}
\label{sec:poisson_sampling}

\begin{lemma}
\label{lem:poisson_conditional}
The conditional distribution of a Poisson point process of intensity $\lambda(x)$ on a domain $K$ is $\frac{1}{\Lambda(K)}\lambda(x)$. 
\end{lemma}

Lemma \ref{lem:poisson_conditional} allows one to sample a homogeneous Poisson point process of rate $\lambda$ by first sampling a Poisson random variable $N_{\lambda\mu(K)}=n$, and scattering the $n$ as points as i.i.d. uniform random variables in $K$. Point processes with a hard-core repulsion (such as the Mat\'ern Type-II point process \cite{StochasticGeometry}) can often be realized as a {\it thinning} of a Poisson point process.

\begin{definition}[Mat\'ern type-II hard-core point process, \cite{StochasticGeometry}]
\label{def:matern_pp}
Let $N$ be a homogeneous Poisson point process of rate $\lambda$. To each point $X_i$ of $N$, assign a {\it mark} $M_i$ which is an i.i.d. uniform random variable on $(0, 1)$. Then, construct the Mat\'ern process $N'$ with hard-core repulsion distance $r$ as
\begin{align}
    \label{eq:matern}
    N' &= \{X_i \in N\ |\ M_i < M_j\ \forall M_j \in B(X_i, r),\ i \ne j\}.
\end{align}
The points of $N'$ are situated at a minimum distance $r$ from one another.
\end{definition}

\subsubsection{Poisson and Mat\'ern processes on $\mathbb{S}^d$}

Recall that a zero-mean multivariate Gaussian $W_{0,\Sigma}$ with $\Sigma$ a $d\times d$ symmetric positive-definite covariance matrix has the ellipsoids $\{x\ |\ x^\intercal \Sigma^{-1}x = r^2\}$ as level sets of constant density. When $\Sigma = I$, this is the property of {\it spherical symmetry}. Combined with the fact that the map $x \mapsto x \norm{x}^{-1}$ is surface-area-preserving up to a constant multiple, this yields a computationally efficient and numerically stable method for generating iid uniform random variables $U_i$ on $\SP^d$: $U_i = W_{i;0,I} \norm{W_{i;0,I}}^{-1}$.

Lemma \ref{lem:poisson_conditional} then allows the realization of homogeneous Poisson processes on $d$-spheres of radius $r$ as $N_{\lambda Ar^d}=n$, $\{U_i\}_{i=1}^n$, with $A$ the surface area of the unit sphere $\SP^d$ and $\lambda$ the intensity per unit area. Conveniently, the parametrization of $U_i$ in Cartesian coordinates allows the realization of Mat\'ern  processes (\ref{def:matern_pp}) using the geodesic (great-circle) distance $d(U_i, U_j) = r\arccos(U_i^\intercal U_j)$.

\section{Results for two-dimensional Poisson-Voronoi tessellations}

In dimension $d\ge 2$, known analytical results concerning Voronoi tessellations of point processes are limited to lower-dimensional facets, such as edge (2D) and face (3D) distributions (see \cite{MollerGamma} and \S 4.4, \cite{MollerVoronoi}). The difficulty in higher dimensions arises primarily from the loss of uniqueness for shapes satisfying geometric properties\textemdash such as fixed measure (length, area, volume)\textemdash for which one must first consider distributions over shapes, then marginalize over the level-sets satisfying a scalar geometric property, such as aspect ratio. For this reason, in the following models we present primarily numerical results (with partial analytical arguments where applicable), and present a validation of the numerical method in \ref{sec:pv_validation}.

\begin{figure*}[t]
\includegraphics[width=\columnwidth]{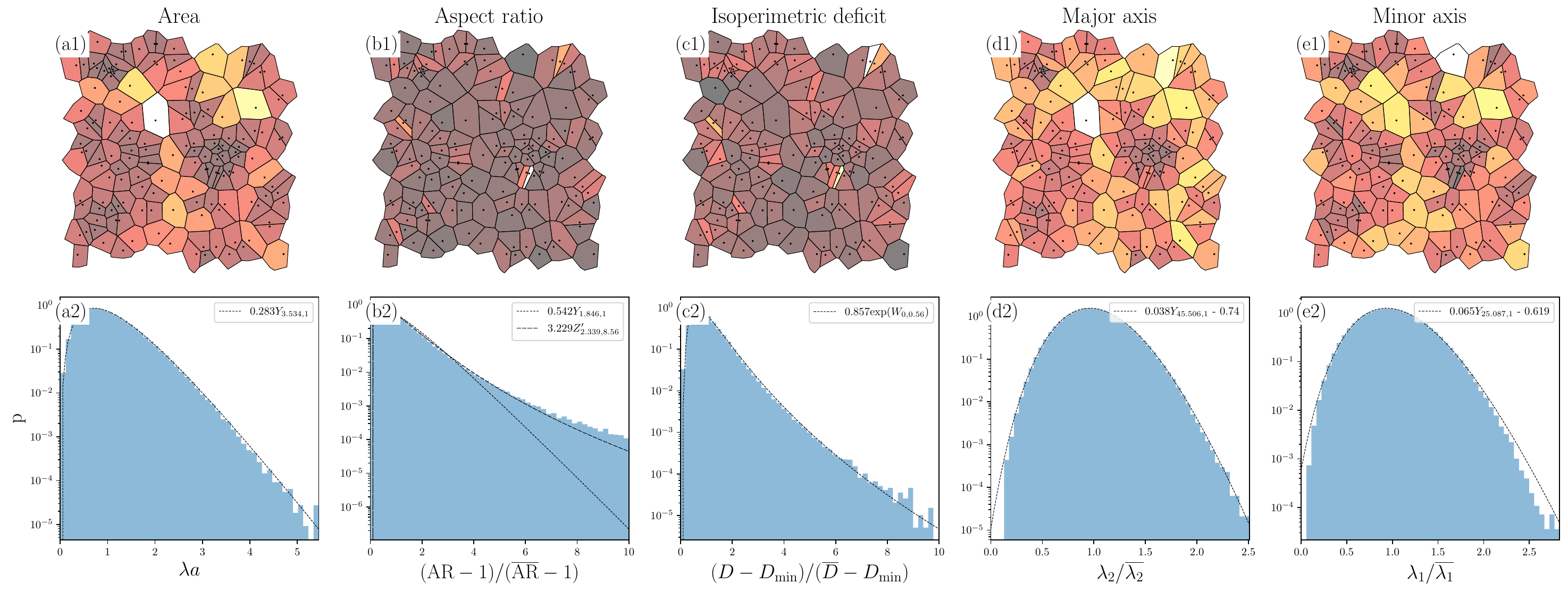}
\caption{Poisson-Voronoi tessellation on the periodic unit square. Panels (a1) - (e1) show the Voronoi cells colored by four geometric  quantities defined from the $n$th moments of area and perimeter. Panels (a2) - (e2) show empirical and maximum-likelihood estimations of gamma (denoted $Y$), lognormal (denoted $\exp(W)$), and beta prime (denoted $Z'$)  distributions where applicable. Poisson-Voronoi tessellations do not have gamma-distributed aspect ratios, instead following an approximate beta-prime law.}
\label{fig:pp_torus}
\end{figure*}

\hypertarget{II.1}{}
\textit{Model II.1 (periodic unit square)}\textemdash 
Consider a homogeneous Poisson point process on the unit 
square $[0, 1]^2$ with periodic boundary conditions, denoted $\mathbb{T}^2$. On general $d$-dimensional domains, the 
point process is specified by an {\it intensity measure} $\lambda(A)$ for subsets $A \subseteq \mathbb{T}^2$ in which 
the count $N_{A;\lambda(A)}$ is a Poisson random variable of rate $\lambda(A)$. 
A {\it homogeneous} Poisson process \textemdash one whose intensity  $\lambda$ is constant on sets $A$ of constant
measure\textemdash is realizable by sampling the total count $N_{;\lambda(\mathbb{T}^2)}$ and assigning the positions 
$\{R_i\}_{i=1}^n$ conditional on $N=n$ as i.i.d. uniform random variables. Figure \ref{fig:pp_torus} shows numerical 
simulations for $k=1000$ trials with intensity $\lambda(\mathbb{T}^2) = 10^3$, which is on the order of the number of somatic cells in {\it Volvox carteri}.

The periodic Voronoi tessellation on $\mathbb{T}^2$, shown in the small-$n$ 
example in Figure \ref{fig:pp_torus}(a1)-(d1), is constructed by copying $\{R_i\}$ in four quadrants around 
$\mathbb{T}^2$ and selecting the sub-tessellation corresponding to the original points. Areas $a$, nondimensionalized 
as $\lambda a$, conform to a gamma random variable with $k\approx 3.5$, as in Fig. \ref{fig:pp_torus}(a2). The 
isoperimetric deficit $D= L / \sqrt{4\pi a} - 1$ 
with $L$ the perimeter, conforms after an appropriate rescaling to a log-normal random variable with $\sigma 
\approx 0.6$, shown in Fig. \ref{fig:pp_torus}(c2). The aspect ratio $\text{AR}$ of a Poisson-Voronoi tessellation 
does not conform to a gamma distribution (Fig. \ref{fig:pp_torus}(b2)), in contrast to confluent tissue \cite{Atia} in which gamma-distributed aspect ratios appear in a diverse range of densely-packed cell types and 
inert matter. 
Instead, AR conforms approximately to a beta-prime distribution, which naturally arises as the 
ratio of independent gamma random variables. Figure \ref{fig:pp_torus}(d2-e2) shows that the major and minor axis 
lengths are approximate gamma-distributed; while we do not in general expect the major and minor axis to be independent 
random variables, the beta-prime distribution governs the aspect ratio in the event that they are, with shape parameter $k_1 / k_2$.


\subsection{Validating Poisson-Voronoi simulations}
\label{sec:pv_validation}

Let $V_i$ be the measure (area, volume, etc.) of the typical Poisson-Voronoi cell.
While the distribution of $V_i$ is presently unknown, numerically integrated second moments of $V_i$ in $\mathbb{R}^2$ and $\mathbb{R}^3$ have been reported \cite{GilbertCrystals}, facilitating comparison with numerical study. Large simulations \cite{TanemuraVoronoi, FerencVoronoi} with $n > 10^6$ cells have found that gamma distributions, and in particular a 3-parameter generalization \cite{TanemuraVoronoi}
\begin{equation}
    \label{eq:generalized_gamma}
    f_{Y_{;k,\lambda,a}}(v) = \frac{a\lambda^{k/a}v^{k-1}e^{-\lambda v^a}}{\Gamma(k/a)}
\end{equation}
achieve good maximum-likelihood fit to data with $<1\%$ error relative to the analytical second moment. 
In Fig. 
\ref{fig:pv_torus_stats}, the estimated second 
moment $\langle V_i^2\rangle$, shape parameter $k$, and 
CDF root mean square error are displayed for $500$ trials 
of $N\pdfsim\text{Poisson}(10^3)$ points. The average 
empirical, gamma, and beta second moments show good 
agreement with Gilbert's \cite{GilbertCrystals} numerically
integrated value of $1.280$ and are within $1\%$ relative 
error, validating the numerical method. The estimated value
of $k\approx3.7$ for gamma on the torus $\mathbb{T}$ is consistent with prior results finding 
$k\approx 3.6$ \cite{WeaireVoronoi} in the plane $\mathbb{R}^2$. On the other hand, the 
estimated value of $k \approx 3.2$ for beta is lower than
gamma and closer to Tanemura's \cite{TanemuraVoronoi}
generalized-gamma \eqref{eq:generalized_gamma} fit finding 
$k=3.315$, suggesting that a beta hypothesis is a good 
substitute for the generalized gamma distribution. Lastly, 
we observe that the beta RMSE is slightly decreased compared 
to gamma.

\subsection{A conjecture for the distribution of Poisson-Voronoi areas}

In dimension $d \ge 2$, exact distributions for the cell measures (areas, volumes, and so on) of Delaunay triangulations (denoted $D_i$) of a Poisson point process are known \cite{RathieDelaunay}, with the distribution in $\mathbb{R}^2$ given by a modified Bessel function of the second kind,
\begin{eqnarray}
    \label{eq:bessel_cell_density}
    f_{D_i}(v) &=& c_1v^{\alpha_1}K_{n}(c_2v^{\alpha_2}),
\end{eqnarray}
with parameters $c_i, \alpha_j, n$. On the other hand, exact distributions for their vertex-cell duals\textemdash the Voronoi tessellations\textemdash are presently unknown. We conjecture, however, that \eqref{eq:bessel_cell_density} should also govern the Voronoi areas $V_i$ based on the following heuristic argument.  
Conditional independence of the volume of the \textit{fundamental region} (a particular set containing the vertices of a Voronoi cell and the origin) and its shape (see \cite{MollerGamma}) suggests that one may assume a particular approximate shape for the typical Voronoi cell, say, an approximately elliptical region.
If, as seen in Fig. \ref{fig:pp_torus}, the principal axes are gamma-distributed, and additionally they are independent, then $V_i$ is proportional to their product and has a distribution precisely of the form \eqref{eq:bessel_cell_density} (see \cite{MalikGamma}).
A possible route to proving this claim is to show that the desired independence holds in a thermodynamic limit.

\begin{figure*}[t]
\centering
\includegraphics[width=\columnwidth]{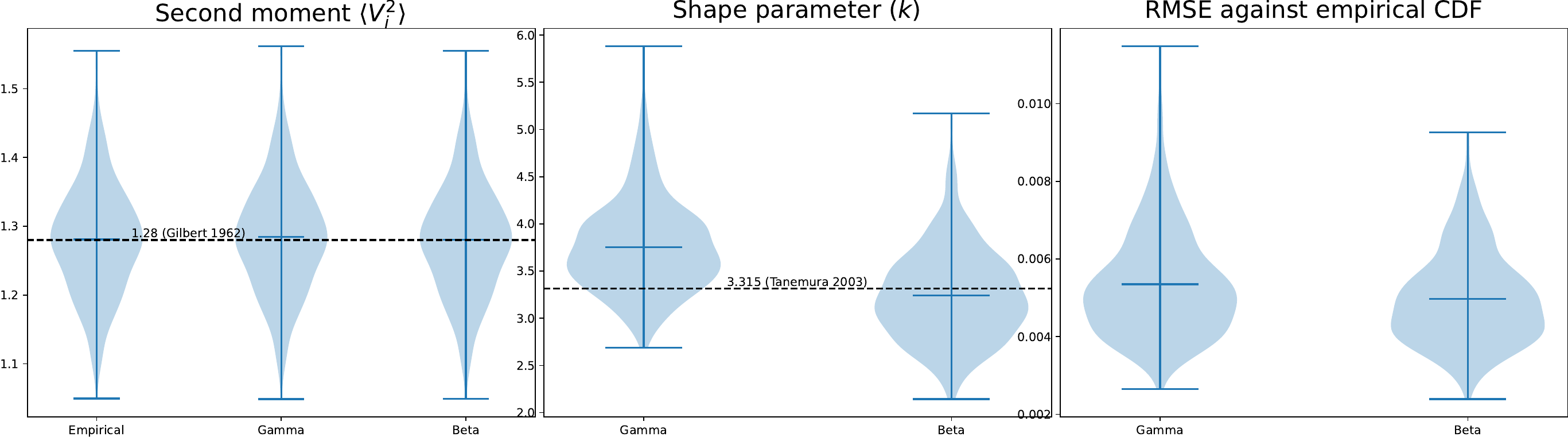}
	\caption{Gamma- and Beta-maximum likelihood fits for Poisson-Voronoi tessellations of the flat torus $\mathbb{T}^2$. The numerically integrated value of the second moment $\mathbb{E}[V_i^2] = 1.280$ found by \cite{GilbertCrystals} is plotted as a dashed line in the left plot, indicating good agreement with our numerical simulations. The experimentally determined shape parameter $k=3.315$ of the generalized-gamma fit found by \cite{TanemuraVoronoi} is plotted as a dashed line in the middle plot.}
	\label{fig:pv_torus_stats}
\end{figure*}


\section{The Voronoi Liquid}
\label{sec:sm_qz}

Let $\Omega \subset \R^d$ be a compact connected measurable set with positive $d$-dimensional Lebesgue measure $\magn{\Omega} > 0$ as in the main text.
For a finite set of points $\{x_i\}_{i=1}^n = \mathcal{X} \subset \Omega^d$ and associated partitions $\{D_i\}_{i=1}^n = \mathcal{D}$ with $D_i\cap D_j = \emptyset$, $\bigcup_i D_i = \TP^d$, the \textit{quantization} problem is to minimize the cost
\begin{align}
    \label{eq:quantizer_1}
    E(\mathcal{X}, \mathcal{D}) &= \sum_{i=1}^n E_i,\quad  E_i = \int_{D_i}\norm{x_i - y}_2^2dy,
\end{align}
which is the approximation error of points $y \in \Omega^d$, and therefore the domain $\Omega^d$ itself, by the partitions $\{x_i, D_i\}$.
This illustrates the origin of \eqref{eq:quantizer_1} in meshing problems \cite{CVT1}.

\subsection{The intensive energy}
Define $U$ to be the  \textit{energy per unit measure}
\begin{align}
    \label{eq:qz3}
    U(\mathcal{X}) &= \frac{n^{2/d}}{\magn{\Omega}^{1+2/d}}E(\mathcal{X}, \text{Vor}(\mathcal{X})).
\end{align}
We argue \eqref{eq:qz3} is an intensive quantity as follows.
First, consider the ``tiling'' (frozen) phase in which  $\mathcal{X}, \mathcal{D}$ is a tiling of $\Omega$ with $x_i$ the centroids of congruent cells $D_0\cong D_i$ and $\magn{\Omega} = n\magn{D_0}$. 
Then \eqref{eq:qz3} becomes
\begin{align}
    U_{\text{tiling}}(\mathcal{X}) &= \frac{\rho^{2/d}}{\magn{\Omega}}\sum_{i=1}^n\int_{D_i}\norm{x_i-y}^2dy = \frac{\rho^{2/d}}{d\magn{D_0}}\tr(S_i(\mu_i)) = \frac{1}{\magn{D_0}^{1+2/d}}\tr(S_i(\mu_i)),
\end{align}
which is simply the cell's second moment $S_i$, nondimensionalized by the cell measure. 
Recall that in 2D, for example, $\tr(S_i)$ has units $\text{length}^4$.
More generally, if $\langle E_i\rangle$ is the average partition energy, then $U$ is expressible as
\begin{align}
    U(\mathcal{X}) &= \rho^{1+2/d}\langle E_i\rangle.
\end{align}
Now we note that $U$ is scale-invariant as follows. 
Let $x \mapsto \alpha x$ with $\alpha > 0$, so that $\magn{\Omega}\mapsto \alpha^d\magn{\Omega}$ and $E_i \mapsto \alpha^{2+d}E_i$. 
Then, 
\begin{align}
    U(\mathcal{X}) &\mapsto \frac{n^{1+2/d}}{\alpha^{d+2}\magn{\Omega}^{1+2/d}}\cdot \alpha^{d+2}\langle E_i\rangle = U(\mathcal{X})
\end{align}
Lastly, let us make an informal argument for $U$ being a dimensionless quantity independent of $n$ in the general-$\mathcal{X}, \mathcal{D}$ case.
Let $a \sim b$ indicate that $a$ is of the same units and scale as  $b$, such that $\magn{\Omega} \sim V$ is the scale of the system measure. 
Then,
\begin{align}
    U(\mathcal{X}) &= \frac{\rho^{2/d}}{\magn{\Omega}} \sum_i E_i \sim \left(\frac{n}{V}\right)^{2/d}\cdot \frac{1}{V}\cdot n \cdot E_i.
\end{align}
Let $\norm{x_i - y} \sim \ell$ be the length scale of the typical partition and note that $E_i \sim \ell^{2+d}$.
Assuming further the decomposition $\ell^d\sim V/n$, the generalization from the tiling case, we have
\begin{align}
    &\sim (\ell^{-d})^{2/d}\cdot \ell^{-d}\cdot \ell^{2+d} \sim 1,
\end{align}
so that $U$ is dimensionless.

\subsection{Interpreting $E$ as a strain energy}
Eq. \eqref{eq:quantizer_1} is expressible in terms of the second area moment $S_i(x_i)$ about $x_i$ as
\begin{align}
    \label{eq:quantizer_m2}
    E(\mathcal{X}, \mathcal{D}) &= \sum_i\tr S_i(x_i).
\end{align}
Since for affine transforms $T_{x}(y) = F(y-x) + x$ about the point $x$ for some $F$ with $\det(F) > 0$, we have that the second moment maps as $S(x) \overset{y\mapsto T(y)}{\mapsto} FS(x)F^\intercal \det F$, and \eqref{eq:quantizer_m2} is further equal to 
\begin{align}
    &= \sum_i m_i \tr F_iF_i^\intercal \det F_i
\end{align}
where $\frac{m_i}{p}I$ is the second area moment of an isotropic natural configuration (e.g. a circle or regular $p$-gon) from which $F_i$ is a strain tensor.
It was already shown in the main text that the optimal $\mathcal{D}$ for fixed $\mathcal{X}$ is the Voronoi tessellation, hence for convex $\Omega$, $D_i$ are convex polygonal.
Then, as argued in \cite{E2} (SM), $\tr( FF^\intercal)$ is the bulk strain energy, in the small-strain limit, of a deformation from a regular $n$-gon, for any isotropic frame-indifferent constitutive relation (which therefore depends only on the principal tensor invariants of $F$).
Note that, for dilatations $\Omega \mapsto \alpha \Omega$ with $\alpha > 0$, 
\begin{align}
    E[\{\alpha x_i, \alpha D_i\}] &\mapsto \alpha^{2+d}E[\{x_i, D_i\}].
\end{align}

\subsection{Gradient flow is a nonlinear diffusion in 1D}
When $d=1$, for positions $x = (x_1,\ldots,x_n) \in O$ where $O \subset \R^n$ are the ordered vectors (also called \textit{Weyl chamber})
\begin{align}
    O &= \{y \in \R^n\ |\ 0 \le y_1 \le \ldots \le y_n < 1 \le y_1 + 1\},
\end{align}
we have the Voronoi centroids and lengths
\begin{align}
    \mu_{i} &= \frac{x_{i-1} + 2x_{i} + x_{i+1}}{4},\quad v_{i} = \frac{x_{i+1} - x_{i-1}}{2},
\end{align}
with $x_{n+1} = x_1$.
Thus the energy \eqref{eq:quantizer_m2} becomes the nearest-neighbors cubic potential
\begin{align}
    \label{eq:quantizer_1d}
    V(\mathcal{X}) &= \sum_i\int_{\frac{x_{i-1}-x_{i}}{2}}^{\frac{x_{i+1}-x_i}{2}}y^2dy = \frac{1}{24}\sum_i\left[(x_{i+1}-x_i)^3 + (x_i - x_{i-1})^3\right].
\end{align}
Moreover, its gradient descent, as we see to be consistent with \eqref{eq:lloyd}, is 
\begin{align}
    \label{eq:lloyd_1d}
    \dot{x}_{i} &= - \frac{\partial U}{\partial x_i} = 2\magn{D_i}(\mu_i - x_i)\\
    &= (x_{i+1}-x_{i-1})\left(\frac{x_{i-1}+2x_i+x_{i+1}}{4} - x_i\right)\\
    &= \frac{1}{4}\left((x_{i+1} - x_{i})^2 - (x_{i} - x_{i-1})^2\right),
\end{align}
whose fixed point is the 1D lattice satisfying $(x_{i+1} - x_{i})^2 = (x_{i} - x_{i-1})^2$ for all $i$.
The dynamics of the spacings $\ell_i = x_{i+1}-x_i$ are then given by
\begin{align}
    \label{eq:elldot}
    \dot{\ell}_i &= \frac{d}{dt}(x_{i+1} - x_i) = \frac{1}{4}(\ell_{i+1}^2 - 2\ell_i^2 + \ell_{i-1}^2) = \frac{1}{4}\sum_{j=i\pm 1} (\ell_j^2 - \ell_i^2).
\end{align} 
We recognize from the final expression that \eqref{eq:elldot} is a centered-differences discretization of the porous medium equation $\partial_t u = \partial_{xx}u^2$.
Moreover, from \eqref{eq:quantizer_1d}, we see that the equivalent energy over lengths $ (\ell_1,\ldots,\ell_n) = \mathcal{L}$ is
\begin{align}
    V(\mathcal{L}) &= \frac{1}{24}\sum_i(\ell_i^3 + \ell_{i-1}^3),
\end{align}
with $\ell_i$ nonnegative and subject to total length constraint $\sum_i\ell_i = L$.
For uniformly distributed (Poisson) initial conditions $x_i$ on $[0,1]$ with periodic boundary conditions, the resulting Voronoi lengths under the dynamics \eqref{eq:elldot} are well-approximated by beta random variables as shown in Figure \ref{fig:quantizer_circle}.

\begin{figure*}[t]
\includegraphics[width=\columnwidth]{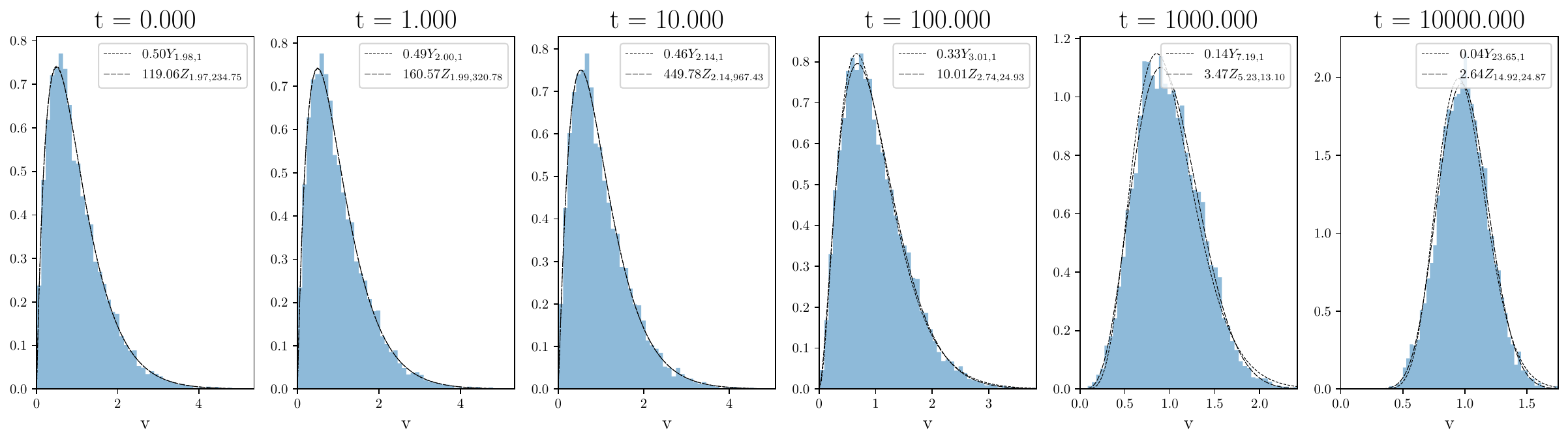}
\caption{Gradient flow of the quantization energy from Poisson initial conditions on the unit circle have Voronoi lengths which are well-approximated by beta and gamma distributions with increasing shape parameter $k$.}
\label{fig:quantizer_circle}
\end{figure*}

\subsection{Stochastic gradient flow on the torus}
Fix now $\Omega = \TP^d$.
It is shown in \cite{CVT1} that $E$ is $C^2$ and 
\begin{align*}
    \frac{\partial U}{\partial x_i} &= \frac{2\magn{D_i}}{\magn{\TP^d}}(x_i - \mu_i) = 2\magn{D_i}(x_i - \mu_i),
\end{align*}
where $\mu_i$ is the centroid of $D_i$.
This suggests the continuous-time gradient flow
\begin{align}
    \label{eq:lloyd}
    \dot{x_i} &= 2\magn{D_i}(\mu_i - x_i).
\end{align}
Consider now the It\^o process
\begin{align}
    \label{eq:stochastic_lloyd}
    dR_i(t) &= \frac{\partial U}{\partial x_i}\biggr\rvert_{R_i(t)} dt + \sigma dW_i(t),
\end{align}
where $W_i(t)$ are independent standard Brownian motions.
In the limit $\sigma\to\infty$, $(\mathcal{X}(t), \mathcal{D}(t))$ converges in distribution (in large $t$) to a Poisson-Voronoi tessellation, while in the limit $\sigma\to 0$ it converges to a centroidal Voronoi tessellation \cite{CVT1, CVT2}. 
The corresponding Fokker-Planck equation for the joint probability density $p(R_1,\ldots,R_n) =: p(R)$ of \eqref{eq:stochastic_lloyd} is
\begin{align}
    \label{eq:lloyd_fp}
    \frac{\partial p}{\partial t} &= -\nabla \cdot p \nabla U + \frac{\sigma^2}{2} \Delta p.
\end{align}
Since $U$ is $C^2$ and now compactly supported, obeying growth conditions required \cite{RiskenFP},
a Gibbs measure $\mu_\beta$ at temperature $\beta^{-1}$ exists and is the stationary solution of \eqref{eq:lloyd_fp}, given by
\begin{align}
    \label{eq:lloyd_gibbs}
    \frac{d\mu_\beta(x)}{dx} &= f_\beta(x) = \frac{1}{Z_\beta}\exp\left(-\beta U(x)\right),
\end{align}
with $Z_\beta$ the normalizing constant and $\beta^{-1} = \sigma^2 / 2$ the diffusion constant. Then \eqref{eq:lloyd_gibbs} minimizes the Gibbs free energy functional
\begin{align}
    G_\beta &= \mathbb{E}_{f_\beta}[U] - \beta^{-1}H[f_\beta],
\end{align}
with $\mathbb{E}_{f_\beta}$ denoting the expectation with respect to $f_\beta$ and $H$ the differential entropy of $f_\beta$.
The previous limits\textemdash Poisson-Voronoi and centroidal Voronoi tessellations\textemdash correspond to the infinite- ($\beta \to 0$) and zero- ($\beta \to \infty$) temperature limits respectively.


\subsection{Specifics of the Langevin simulation}
Let us now sample from the Gibbs measure \eqref{eq:lloyd_gibbs} using Markov Chain Monte Carlo (MCMC).
We take $\Omega = \TP^d$ to be the unit $d$-cube with periodic boundary conditions.
Euler-Maruyama integration of (8) with discrete steps $\varepsilon > 0$ 
is performed by 
\begin{align}
    \label{eq:langevin_mc}
    R_i(t+\varepsilon) &= -\varepsilon \frac{\partial V}{\partial x_i}\evalat_{R_i(t)} + \sqrt{2\varepsilon\beta^{-1}} \eta,\ \eta \pdfsim \mathcal{N}(0, 1).
\end{align}
As \eqref{eq:langevin_mc} no longer necessarily satisfies detailed balance with respect to (7), we add a Metropolis-Hastings step, a procedure known as the Metropolis-adjusted Langevin algorithm (MALA) \cite{RHMC}. 
That is, the transition $R_i(t) \mapsto R_i(t+\varepsilon)$ occurs with probability $\min(1, \alpha)$, where
\begin{align}
    \label{eq:mala}
    \alpha &= \exp\left(-V_\Delta - \beta\frac{\norm{\varepsilon \nabla V_{t+\varepsilon} - R_\Delta}_2^2 - \norm{\varepsilon \nabla V_t + R_\Delta}_2^2}{4\varepsilon}\right)
\end{align}
and $V_\Delta = V_{t+\varepsilon}-V_t$ and $R_\Delta = (R_i(t+\varepsilon) - R_i(t))_{i=1}^n$. 
We iterate \eqref{eq:langevin_mc}, \eqref{eq:mala} for $n=1000$ particles in the unit square with periodic boundary conditions.
For non-Poisson configurations ($\beta > 0$, Fig. 2(b)-(d)), we use the ``frozen'' initial condition (Fig. 2(d)), computed using a gradient descent of (5), which has a faster mixing time than starting with the infinite-temperature (Poisson) initial condition. 

We use the following simulation conditions to produce Fig. 2 in the main text. 
\begin{itemize}
    \item $n = 1000$ particles
    \item $k = 1000$ samples at each temperature $\beta^{-1}$ in distinct Markov chains
    \item $\beta \in [10^{-3}, 1]$ at 13 logarithmically spaced intervals and $\beta \in [1, 40]$ at 21 linearly spaced values
    \item $n_{\text{mix}} = 1000$ mixing steps for each chain prior to sampling
    \item $\varepsilon = 1.0$ time step size
    \item $n_{\text{freeze}} = 2000$ gradient steps of size $\varepsilon_{\text{freeze}} = 0.1$ to produce the frozen initial condition
\end{itemize} 
The sufficiency of the mixing time $n_{\text{mix}}$ is validated by running the same procedure again with double the time $2 n_{\text{mix}}$, which produced no measurable differences in the resulting distributions except at high $\beta=100, k\approx 2000$.

\subsection{Enforcing hard-sphere constraints}

\begin{figure*}[t]
\includegraphics[width=\columnwidth]{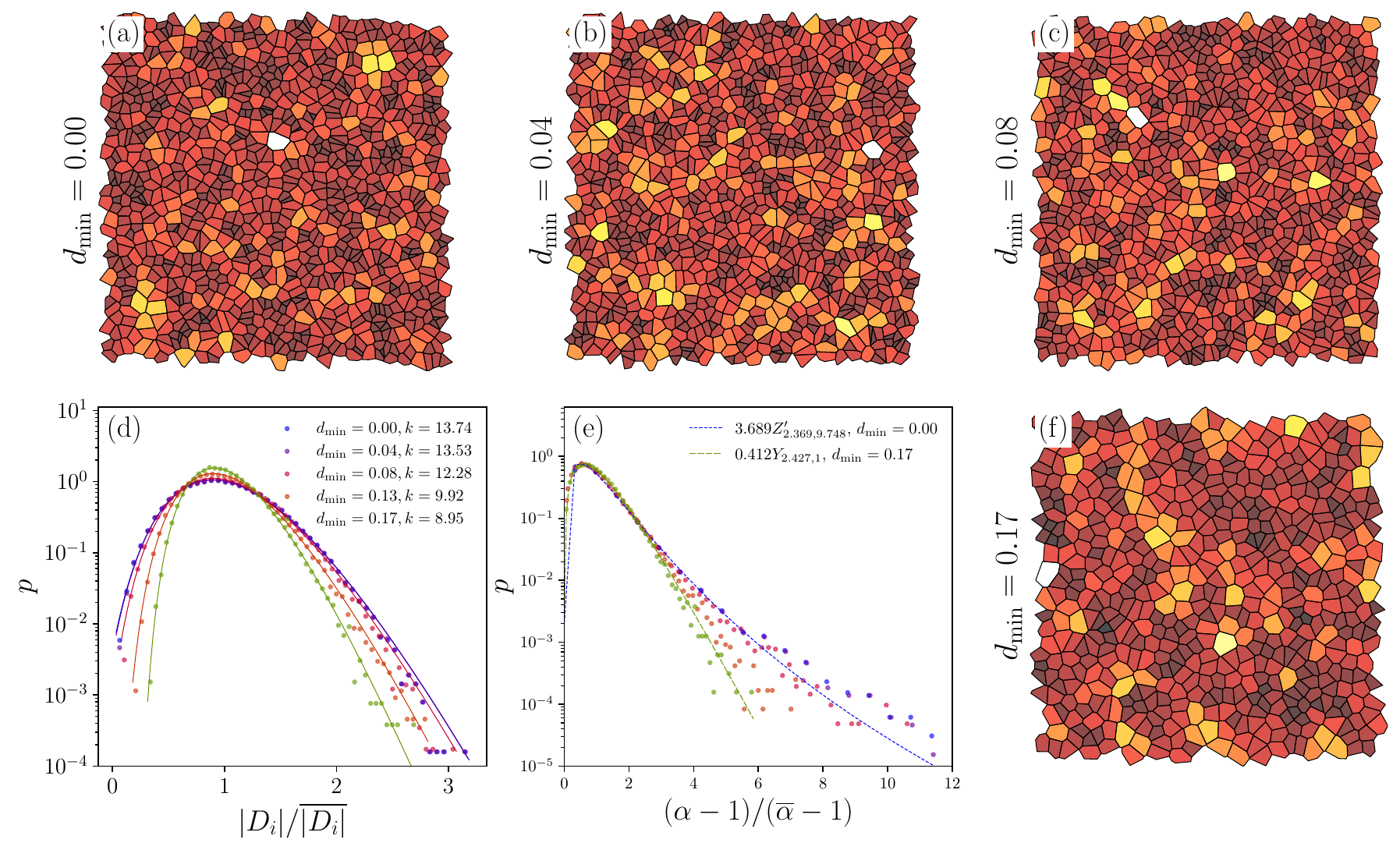}
\caption{Transition to gamma-distributed aspect ratios under Mat\'ern thinning of the Voronoi liquid at fixed $\beta = 10^4$. Panels (a-c, f) show samples thinned by the Mat\'ern rule \eqref{eq:matern} at various distances $d_{\min}$. 
Panel (d) shows that partition areas are well-described by $k$-gamma distributions at all $d_{\min}$, with $k$ monotone decreasing. 
Panel (e) shows that partition aspect ratios, calculated from the second area moment, transition from an approximate beta-prime distribution to a gamma distribution with $k \approx 2.4$.}
\label{fig:vliq_hc}
\end{figure*}

We now apply a hard-sphere constraint to samples of the fixed-$N$, $V$ ensemble (7). 
One of the simplest methods is the Mat\'ern thinning rule \eqref{eq:matern}, which filters points from a particular configuration $\mathcal{X}$ to $\mathcal{X}'$ such that a minimum user-specified spacing $d_{\min}$ is satisfied.
Fig. \ref{fig:vliq_hc} shows several such samples at fixed $\beta = 10^4$ and $d_{\text{min}}$ a multiple $\alpha \in \{0, 0.25, 0.5, 1\}$ of the minimum spacing in the frozen phase (Fig 2(d) in the main text).
As seen in Fig. \ref{fig:vliq_hc}(e), the hard-sphere thinning shows a transition to gamma-distributed aspect ratios.

A second method is to add a divergent term to $E$ (5) for invalid configurations for which $f_\beta$ (7) is vanishing. 
The Langevin diffusion (8) is of course no longer interpretable due to the loss of derivatives. 
However, we can interpret its discretization \eqref{eq:langevin_mc} at valid configurations as an arbitrary sequence of proposals (still outperforming the random walk) for which the Metropolis-Hastings step \eqref{eq:mala} rejects invalid configurations ($V_\Delta = \infty$) and ensures detailed balance with respect to $f_\beta$. 
We mention this approach for completeness; however, additional modifications to MALA need to be made to improve the sampling efficiency, as we find this procedure to be computationally intractable at even moderate hard-sphere radii.
One possible approach is to use Hamiltonian Monte Carlo \cite{RHMC} with constraints.

\section{Approximate Voronoi tessellation of the surface of \textit{V. carteri}}

Below are standard data processing procedures recalled for a self-contained reference, which we compose in a pipeline to produce a simplicial Voronoi tessellation about the somatic cells on the surface of \textit{Volvox}.

\subsection{Fitting ellipsoids in $d$ dimensions}
\label{sec:ellipsoid}
Let $x, v$ be vectors in some $d$-dimensional basis.
The equation of a $(d-1)$-sphere $S$ centered at $v$ is
\begin{align}
    \label{eq:n_sphere}
    \norm{x-v}_2^2 &= 1\quad \forall x\in S.
\end{align}
Applying a rotation $P$ (an orthogonal matrix) of the sphere onto a set of principal axes and stretching along those axes by $\Lambda$ (a positive diagonal matrix), one generalizes \eqref{eq:n_sphere} to an ellipsoid $E$ via a symmetric positive-definite matrix $M = P\Lambda P^\intercal$ defining a generalized inner product in which $x$ satisfies
\begin{align}
    \label{eq:ellipsoid}
    (x-v)^\intercal M(x-v) &= 1\quad \forall x \in E.
\end{align}
The elliptic radii $r_i$ are then $r_i = \Lambda_{ii}^{-\frac12}$ and the elliptic axes are the columns of $P$.

\subsubsection{The minimum volume bounding ellipsoid} 
\label{sec:ellipsoid_bounding}
Given the volume
\begin{equation}
    V = \frac{\pi^{d/2}}{\Gamma(d/2+1)} \prod_{i=1}^d r_i \propto \sqrt{\det(M^{-1})},
\end{equation}
it follows that maximizing $\log\det(M)$ minimizes $V$. Hence for a given dataset $\{x_i\}_{i=1}^n$, the following convex program computes the minimum-volume bounding ellipsoid:
\begin{align}
    \sup_{M,v}\ & \log\det(M)\\
    \text{subject to }\ & (x_i-v)^\intercal M(x-v_i) \le 1\quad \forall i\\
    & M > 0.
\end{align}
Here, $M > 0$ in the sense of linear matrix inequality (LMI), i.e. $M$ is constrained to lie in the positive-definite cone.
However, the offset $x_i-v$ produces a variable-product constraint which is not disciplined convex (DCP). The following reparametrization uses the invertibility of $M$ to convert the constraint to linear least squares:
\begin{align}
    \sup_{A,b}\ & \log\det(A)\\
    \text{subject to }\ & \norm{Ax_i - b}_2^2 \le 1\quad \forall i\\
    A &> 0\\
    M &= A^2 = A^\intercal A\\
    v &= A^{-1}b.
\end{align}
DCP solvers such as CVXOPT \cite{CVXPY} solve this problem. We further reduce the problem size by considering the subset of $X$ lying on the convex hull, computed in $O(n\log n)$.

\subsubsection{$\ell^2$-minimal projection to ellipsoids} 
\label{sec:ellipsoid_l2}
Given a representation of an ellipsoid as $(M, v)$ in the same basis as $x$, define the following convex program:
\begin{align}
    \inf_{Y}\ & \norm{Y - X}_2^2\\
    \label{eq:ellipse_constr}
    \text{subject to }\ & \norm{A(x_i-v)}_2^2 = 1\quad \forall i\\
    & A^\intercal A = A^2 = M,
\end{align}
with $A$ computed by (Hermitian) eigendecomposition of $M$.
Then $Y$ is the minimum-$\ell^2$-distance projection of $X$ onto $(M, v)$. This problem is not DCP; however, since there are no matrix cone constraints, we can simply use non-DCP solvers compatible with nonlinear constraints, such as SLSQP \cite{SLSQP}. The constraint Jacobian for \eqref{eq:ellipse_constr} is $2M(x_i-v)$. 

\begin{figure*}[t]
\centering
\includegraphics[width=\columnwidth]{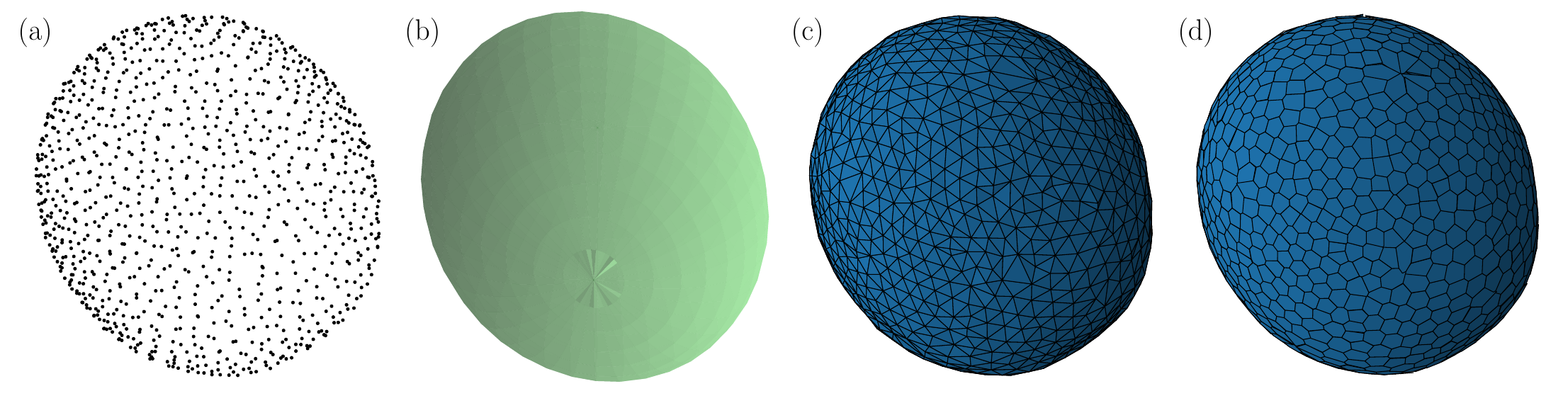}
	\caption{Computational pipeline for tessellating the surface of {\it V. carteri.} (a) The somatic cell positions, available from \cite{Dayetal}, displayed in three-dimensional space. (b) The minimum-volume bounding ellipsoid containing the somatic cells is computed. This establishes the approximate anterior-posterior axis of the spheroid. (c) The $\ell^2$-optimal projection of the somatic cell positions $X$ to this ellipsoid are computed, and the Delaunay tetrahedralization of this projected point cloud $\hat{X}$ is computed. As $\hat{X}$ is now its own convex hull, the index triplets corresponding to the triangular faces lying on the convex hull of the tetrahedral complex are taken to be an approximate Delaunay triangulation of the original point cloud $X$. (d) A quasi-2D Voronoi tessellation of the surface is computed by taking the $\ell^2$-optimal planar embedding of the simplicial ring around any vertex, and constructing a planar Voronoi face in that embedding using the usual circumcenter rule. If invalid (self-intersecting) polygons are produced by the planar embedding, they are corrected by taking the 2D convex hull, which in the limit of zero curvature is a no-op. This procedure introduces overlapping artifacts near gonidia (as seen) but is otherwise unaffected by the global radius of curvature which is large compared to the size of individual polygons.}
	\label{fig:volvox_processing}
\end{figure*}

\subsection{Fitting hyperplanes in $d$ dimensions}
\label{sec:planes}

The equation of a hyperplane $H$ in $d$ dimensions is
\begin{align}
    n \cdot (x - v) &= 0\quad \forall x \in H
\end{align}
for $n, v \in \mathbb{R}^d$. Without loss of generality, we may assume that $\norm{n} = 1$, so that $n$ is a unit normal to $H$; expressing $H$ in the form
\begin{align}
    n\cdot x - b &= 0,
\end{align}
we see that $b = n\cdot v$ is the distance from the origin to $H$. It further follows that the distance from an arbitrary point $y\in \R^d$ to the plane is
\begin{align}
    \label{eq:point_plane_dist}
    d(y, H) &= \magn{n\cdot (y - v)} = \magn{n \cdot y - b}.
\end{align}
Now, let $\{x_i\}_{i=1}^N = X \in \R^{N\times d}$ be a set of data points with $N \ge 3$. The best-fit hyperplane (in the $\ell^2$ sense) is
\begin{align}
    \inf_{n,v}\ &\ \norm{(X-v)n}_2^2\\
    \text{subject to}\ &\ \norm{n}_2 = 1.
\end{align}
The cost is bi-convex in the parameters $v$ and $n$. 
For fixed $n$, the critical point of the cost in $v$ is given by
\begin{align}
    0 &= \frac{\partial}{\partial v}\norm{(X-v)n}_2^2 = 2\sum_{i=1}^N(x_i - v).
\end{align}
Then $v = \frac1N\sum_{i=1}^Nx_i$ is the centroid. 
Letting $\overline{X} = X - v$, the critical point of the cost in $n$ is
\begin{align}
    0 &= \frac{\partial}{\partial n}\norm{\overline{X}n}_2^2 = 2\overline{X}^\intercal \overline{X}n.
\end{align}
A common method to estimate $n$ in this LSQ problem is the singular vector of $\overline{X}$ corresponding to the smallest singular value, which of course is $0$ if $\{x_i\}$ are coplanar. 
By \eqref{eq:point_plane_dist} it follows that the $\ell^2$-orthogonal projection of $X$ onto $H$ is
\begin{align}
    Y &= X - (X-v)nn^\intercal.
\end{align}
Let $u \in \R^d$ be a random vector such that $u\times n \ne 0$ (e.g. a Gaussian vector, for which this is almost surely true); then an invertible planar embedding $Y_H$ of a point $Y$ lying in $H$ is defined by the random orthonormal plane basis $B$:
\begin{align}
    B &= \frac{1}{\norm{v}}\begin{bmatrix} v & v\times n \end{bmatrix},\quad v = u\times n,\\
    Y_H &= YB.
\end{align}

\subsection{3D reconstruction \& analysis of {\it Volvox}}

First, 3D meshes and an ellipsoidal approximation of the organism's surface are constructed using the procedure detailed in Figure \ref{fig:volvox_processing}. 
Then, all Voronoi polygon areas (including those near gonidia, which typically introduce high-aspect-ratio polygons) are converted to solid-angles ($4\pi$ times the area fraction of total) and are filtered by the cutoff $v_c = 0.007$ specified in \cite{Dayetal}. The empirical area distribution for each organism is shifted by this cutoff and nondimensionalized to empirical mean $1$, at which point they are combined across all 6 organisms and shown as a combined histogram in Figure 1.

\end{widetext}

\end{document}